\documentclass[11pt,twoside]{article}

\usepackage{amsfonts}
\usepackage{amsmath}
\usepackage{amsthm}
\usepackage{amssymb}
\usepackage{amscd}
\usepackage{eucal}
\usepackage{bbm}
\usepackage{dsfont}

\usepackage{graphicx}

\usepackage{a4}
\usepackage{microtype}
\usepackage{times}
\usepackage{cite}
\usepackage{float}

\usepackage{color}

\let\a=\alpha \let\be=\beta \let\g=\gamma \let\de=\delta
\let\e=\varepsilon   \let\th=\theta
\let\eps=\epsilon
 \let\k=\kappa \let\la=\lambda \let\m=\mu
  \let\p=\pi  \let\s=\sigma
 
\let\ph=\varphi \let\Ph=\phi \let\PH=\Phi \let\Ps=\Psi
\let\Om=\Omega  
 \let\G=\Gamma

\let\qd=\quad \let\qqd=\qquad 

\def\epp{\, .}
\def\epc{\, ,}

\def\tst#1{{\textstyle #1}}
\def\dst#1{{\displaystyle #1}}

\theoremstyle{plain}
\newtheorem{theorem}{Theorem}
\newtheorem*{theorem*}{Theorem}
\newtheorem{lemma}{Lemma}
\newtheorem*{lemma*}{Lemma}
\newtheorem{proposition}{Proposition}
\newtheorem{corollary}{Corollary}
\newtheorem*{corollary*}{Corollary}

\newtheorem*{conjecture*}{Conjecture}

\theoremstyle{definition}
\newtheorem*{definition}{Definition}
\newtheorem*{remark}{Remark}
\newtheorem*{question*}{Question}

\def\2{\frac{1}{2}} \def\4{\frac{1}{4}}

\def\6{\partial}

\def\+{\dagger}

\def\<{\langle} \def\>{\rangle}

\let\nodoti\i
\def\i{{\rm i}}

\def\rd{{\rm d}}
\def\re{{\rm e}}

\DeclareMathOperator{\sh}{sh}
\DeclareMathOperator{\ch}{ch}
\DeclareMathOperator{\tgh}{th}

\DeclareMathOperator{\cth}{cth}

\DeclareMathOperator{\arch}{arch}

\DeclareMathOperator{\arctg}{arctg}

\DeclareMathOperator{\tr}{tr}

\DeclareMathOperator{\one}{\mathds{1}}
\DeclareMathOperator{\Int}{Int}
\DeclareMathOperator{\Ext}{Ext}

\DeclareMathOperator{\sign}{sign}

\DeclareMathOperator{\id}{id}

\def\Re{{\rm Re\,}} \def\Im{{\rm Im\,}}

\def\vv{\mathbf{v}}

\def\xv{\mathbf{x}}

\def\Cv{\mathbf{C}}
\def\Dv{\mathbf{D}}
\def\Ev{\mathbf{E}}
\def\Fv{\mathbf{F}}

\def\fb{\mathfrak{b}}

\def\fz{\mathfrak{z}}

\renewcommand{\appendix}{%
   \renewcommand{\section}{
        \secdef\Appendix\sAppendix}%
   \setcounter{section}{0}%
   \renewcommand{\thesection}{\Alph{section}}%
   \renewcommand{\theequation}{\thesection.\arabic{equation}}%
}

\newcommand{\Appendix}[2][?]{%
     \refstepcounter{section}%
     \setcounter{equation}{0}%
     \addcontentsline{toc}{appendix}%
          {\protect\numberline{\appendixname~\thesection} #1}%
     \vspace{\baselineskip}%
     {\noindent\large\bfseries\appendixname\ \thesection: #2\par}%
     \sectionmark{#1}\vspace{\baselineskip}}

\newcommand{\sAppendix}[1]{%
     {\noindent\large\bfseries\appendixname\:: #1\par}%
     \sectionmark{#1}\vspace{\baselineskip}}

\allowdisplaybreaks

\begin{document}

\thispagestyle{empty}

\begin{center}

{\Large \bf
High-temperature analysis of the transverse dynamical two-point
correlation function of the XX quantum-spin chain}

\vspace{10mm}

{\large
Frank G\"{o}hmann,$^\dagger$
Karol K. Kozlowski$^\ast$ and Junji Suzuki$^\ddagger$}\\[3.5ex]
$^\dagger$Fakult\"at f\"ur Mathematik und Naturwissenschaften,\\
Bergische Universit\"at Wuppertal,
42097 Wuppertal, Germany\\[1.0ex]
$^\ast$Univ Lyon, ENS de Lyon, Univ Claude Bernard,\\ CNRS,
Laboratoire de Physique, F-69342 Lyon, France\\[1.0ex]
$^\ddagger$Department of Physics, Faculty of Science, Shizuoka University,\\
Ohya 836, Suruga, Shizuoka, Japan

\vspace{25mm}

{\large {\bf Abstract}}

\end{center}

\begin{list}{}{\addtolength{\rightmargin}{9mm}
               \addtolength{\topsep}{-5mm}}
\item
We analyse the transverse dynamical two-point correlation function
of the XX chain by means of a thermal form factor series. The series
is rewritten in terms of the resolvent and the Fredholm determinant
of an integrable integral operator. This connects it with a matrix
Riemann-Hilbert problem. We express the correlation function
in terms of the solution of the matrix Riemann-Hilbert problem.
The matrix Riemann-Hilbert problem is then solved asymptotically
in the high-temperature limit. This allows us to obtain the leading
high-temperature contribution to the two-point correlation function
at any fixed space-time separation.
\end{list}

\clearpage

\section{Introduction}
In our recent work \cite{GKKKS17} we have developed an approach
to the calculation of dynamical correlation functions of Yang-Baxter
integrable lattice systems in equilibrium with a heat bath of
temperature $T$. The basic idea was to combine a certain lattice
realisation of a path integral for finite temperature dynamical
correlation functions \cite{Sakai07} with a thermal form factor
expansion introduced in \cite{DGK13a}. As a result we obtained
a thermal form factor series for the dynamical two-point correlation
functions in this class of systems.

A basic example, for which we worked out the series explicitly, is
the transverse correlation function of the XX chain. The XX chain
is a spin-$\2$ model with Hamiltonian
\begin{equation} \label{hxx}
     H_L = J \sum_{j = 1}^L \bigl( \s_{j-1}^x \s_j^x + \s_{j-1}^y \s_j^y \bigr)
           - \frac{h}{2} \sum_{j=1}^L \s_j^z \epp
\end{equation}
Here the $\s_j^\a$, $\a = x, y, z$, are Pauli matrices acting on site
$j \in \{1, \dots, L\}$ of an $L$-site lattice, and periodic boundary
conditions, $\s_0^\a = \s_L^\a$, are implied. The parameters $J > 0$
and $h > 0$ are the strengths of the exchange interaction and of the
external magnetic field.

The XX quantum-spin chain is particularly simple as an integrable model
in that the derivative of the bare scattering phase in its Bethe Ansatz
solution vanishes identically. It is also special since it maps to a
model of non-interacting Fermions by means of a Jordan-Wigner transformation
\cite{LSM61}. For these reasons rather much, in comparison with more
generic integrable models, is known about its correlation functions.
The longitudinal dynamical two-point function at finite temperature was
calculated by Niemeijer \cite{Niemeijer67} using the mapping to Fermions.
We reproduced his result in a neat form from our thermal form factor
series \cite{GKKKS17}. The transverse two-point function, defined as
\begin{equation} \label{defcorrmp}
     \bigl\< \s_1^- \s_{m+1}^+ (t) \bigr\>_T =
        \lim_{L \rightarrow + \infty}
	\frac{\tr \{ \re^{- H_L/T} \s_1^- \re^{\i H_L t} \s_{m+1}^+ \re^{- \i H_L t} \}}
	     {\tr \{ \re^{- H_L/T}\}} \epc
\end{equation}
where $t$ is the time variable and $\s^\pm = \2 (\s^x \pm \i \s^y)$,
is also well studied, but is much harder to access within the free
Fermion approach. In fact, the finite-temperature analysis of (\ref{defcorrmp})
based on a Jordan-Wigner transformation stayed limited to the high-temperature
asymptotics at short distances. Brandt and Jacoby \cite{BrJa76} proved
the rather well-known formula
\begin{equation} \label{brandtjacoby}
     \lim_{T \rightarrow + \infty} \bigl\< \s_1^- \s_1^+ (t) \bigr\>_T =
        \2 \re^{- \i h t - 4 J^2 t^2}
\end{equation}
for the auto-correlation function at infinite temperature. This was
confirmed by Capel and Perk using a different method \cite{CaPe77}.
The same authors then obtained the first few terms in the high-temperature
expansion of the auto-correlation function and of the nearest and
next-to-nearest neighbour correlation functions at $h = 0$ \cite{PeCa77}.

Deeper insight into the asymptotic behaviour of the transverse correlation
function resulted from the study of different Fredholm determinant
representations. The special case of vanishing magnetic field, $h = 0$ in
(\ref{hxx}), maps to the critical transverse-field Ising chain
\cite{MuSh84}, for which a Fredholm determinant representation of
the auto-correlation function was obtained in \cite{MPS83}. Based on this
Fredholm determinant representation the leading long-time asymptotic behaviour
of the transverse auto-correlation function at finite temperature was computed
in \cite{DeZh94b}. A Bethe Ansatz analysis of (\ref{defcorrmp}) was initiated
by Colomo et al.\ in \cite{CIKT92}, where a Fredholm determinant representation
of the correlation function at finite magnetic field was derived. This
Fredholm determinant representation was then used for a long-time,
large-distance asymptotic analysis of the correlation function at fixed
temperature by Its et al.\ \cite{IIKS93b}, who obtained the leading
exponential term and the next-to-leading algebraic corrections. However,
the constant term in the long-time, large-distance asymptotics in the
critical regime $0 \le h < 4J$ has never been calculated.

With our novel form factor series representation \cite{GKKKS17} we have
the opportunity to revisit the problem of an efficient calculation of 
the transverse dynamical two-point function (\ref{defcorrmp}). The
series is based on form factors of the quantum transfer matrix. It
differs from the series employed by Its et al.\ in the derivation of
their Fredholm determinant representation \cite{IIKS93b} of the
correlation function. It gives us direct access \cite{GKS19bpp}
to the constant factor in the long-time, large-distance asymptotics
of the correlation function in the so-called spacelike regime, where
the spatial separation of two space-time points in appropriate units
is larger than the time separation. In fact, the first term in the
series determines the asymptotics in the spacelike regime. As such
the series exhibits a striking similarity with the Borodin-Okounkov,
Geronimo-Case formula \cite{BoOk00,GeCa79,BaWi00} for a Toeplitz
determinant generated by a symbol satisfying the hypotheses of the
Szeg\"o theorem.

For further asymptotic analysis we shall identify our series as
being proportional to the Fredholm determinant of an integral
operator of integrable type. In separate work \cite{GKSS19}
we shall show that this Fredholm determinant representation
is highly efficient for the actual numerical calculation of the
correlation function (\ref{defcorrmp}) in the critical as well
as in the massive regime for generic values of distance, time
and temperature.

In this work we will derive a matrix Riemann-Hilbert problem
associated with our Fredholm determinant representation and express
the two-point function (\ref{defcorrmp}) explicitly in terms of
its solution. This matrix Riemann-Hilbert problem can be used to
calculate the long-time, large-distance asymptotics of the correlation
function in the timelike regime. As we shall see below, it can also
be used to derive the high-temperature asymptotics of the two-point
function for any fixed spatio-temporal separation. This is the main
purpose of this work. We shall obtain a generalization of the
Brandt-Jacoby formula (\ref{brandtjacoby}) to any spatial separation
of points. The generalization, stated more precisely in
Theorem~\ref{theo:main} below, is of the form
\begin{multline} \label{mainresult}
     \bigl\<\s_1^- \s_{m+1}^+ (t)\bigr\>_T = \2 \biggl(- \frac{J}{T}\biggr)^m
        \exp\bigg\{\i c \tau  - \frac{\tau^2}{4}
	           + \int_0^\tau \rd \tau' \: u_m (\tau') \biggr\} \\[1ex] \times
	\frac{Q_{m+1}(- \i) P_m' (- \i) - P_m (- \i) Q_{m+1}' (- \i)}
	     {\bigl(Q_{m+1}(- \i) P_m' (- \i)
	            - P_m (- \i) Q_{m+1}' (- \i)\bigr)\bigr|_{\tau = 0}}
        \bigl(1 + {\cal O} (T^{-2}) \bigr) \epp
\end{multline}
Here $\tau = - 4J\bigl(t - \frac{\i}{2T}\bigr)$, $c = \frac{h}{4J}$,
and $P_m$ and $Q_{m+1}$ are polynomials whose coefficients depend on
$\tau$. These coefficients are explicit but complicated rational
combinations of modified Bessel functions. The same is true for the
function $u_m$.

The paper is organized as follows. In Section~\ref{sec:ffseries} we recall
and slightly rewrite the thermal form-factor series obtained in
\cite{GKKKS17}. In Section~\ref{sec:freddetrep} we recast the series in
the form of a Fredholm-determinant representation. In Section~\ref{sec:intop}
we work out the `integrable structure' of the corresponding integral
operator. Having in mind possible extensions to the more general XXZ chain
we use a parameterization in terms of rapidity variables, entering
the expression for the integration kernel through hyperbolic functions.
In order to be self-contained we work out some of the basic features
of the associated matrix Riemann-Hilbert problem in Section~\ref{sec:mRHp}.
In Section~\ref{sec:corrmRHp} we express the transverse correlation function
in terms of the solution of the matrix Riemann-Hilbert problem. In
Section~\ref{sec:standardform} we transform the matrix Riemann-Hilbert
problem to a certain standard form which, in Section~\ref{sec:hightmrhp},
is asymptotically analysed in the high-temperature limit.
Section~\ref{sec:corrhighT} is devoted to the derivation of our main
result, the high-temperature asymptotic formula (\ref{mainresult}).
A short summary and conclusions are presented in Section~\ref{sec:conclusions}.
The appendices are devoted to the solution of the model Riemann-Hilbert
problem appearing in the high-temperature asymptotic analysis, to
working out some of the properties of the polynomials $P_m$ and $Q_{m+1}$
arising from this analysis, and to the presentation of examples of explicit
high-temperature asymptotic formulae for small $m$.

\section{Thermal form factor series} \label{sec:ffseries}
We start our analysis by recalling and slightly rewriting the
thermal form factor series for the two-point function (\ref{defcorrmp}).
For this purpose we have to introduce a number of basic functions.
\subsection{One-particle energy and momentum}
First of all we define the one-particle energy and momentum
functions. The one-particle momentum as a function of the rapidity
variable is defined by
\begin{equation}
     p(\la) = - \i \ln \bigl( - \i \tgh(\la) \bigr) \epp
\end{equation}
Here we may interpret the logarithm as its principal branch,
meaning that we provide cuts in the complex plane from
$- \i \p/2$ to zero modulo $\i \p$. Below we shall often encounter
the derivative of the momentum function, most conveniently
expressed as
\begin{equation} \label{pprimee}
     \i p' (\la) = \frac{2}{\sh(2 \la)} \epp
\end{equation}
With this the one-particle energy can be defined as
\begin{equation}
     \eps (\la) = h + 2 J p' (\la) \epc
\end{equation}
where $h$ is the magnetic field and $J > 0$ the exchange energy.
In the following we restrict ourselves to the critical parameter
regime
\begin{equation} \label{critreg}
     0 < h < 4J \epp
\end{equation}

Because of the $\p \i$-periodicity of the momentum, which is
shared by all other functions in our form factor series, we
may restrict ourselves to the `fundamental strip'
\begin{equation}
     {\cal S} = \Bigl\{ \la \in {\mathbb C} \Big|
                        - \frac{\p}{4} \le \Im \la < \frac{3 \p}{4} \Bigr\}
\end{equation}
or rather think of the functions as being defined on a cylinder
of circumference $\p$.

It is easy to see that $\eps$ has precisely two roots
\begin{equation}
     \la_F^\pm = \frac{\i \p}{4} \pm z_F \epc \qd
           z_F = \2 \arch \biggl( \frac{4 J}{h} \biggr)
\end{equation}
in ${\cal S}$. These roots are called the Fermi rapidities.
The value
\begin{equation} \label{defpf}
     p_F = p(\la_F^-) = \arccos \biggl( \frac{h}{4J} \biggr)
\end{equation}
of the momentum at the left Fermi rapidity will be called the
Fermi momentum. Using the Fermi rapidities we may rewrite the
one-particle energy as
\begin{equation} \label{epsfactor}
     \eps(\la) = - h \, p'(\la) \sh(\la - \la_F^-) \sh(\la - \la_F^+) \epp
\end{equation}

The functions $\eps$ and $p$ are real on the lines $x \pm \i \p/4$,
$x \in {\mathbb R}$, where they take the values
\begin{subequations}
\begin{align}
     & \eps(x \pm \i \p/4) = h \mp \frac{4J}{\ch (2 x)} \epc \\[1ex]
     & p(x + \i \p/4) = - \frac{\p}{2} + 2 \arctg \bigl( \re^{- 2 x} \bigr) \epc \\
     & p(x - \i \p/4) = - \p \sign(x) + \frac{\p}{2} - 2 \arctg \bigl( \re^{-2 x} \bigr) \epp
\end{align}
\end{subequations}

\subsection{More functions appearing in the form factor series}
In order to define the general term in the form factor series we have to
introduce a few more functions. The one-particle energy determines the function
\begin{equation} \label{funz}
     z(\la) = \frac{1}{2 \p \i} \ln \biggl[ \cth \biggl(\frac{\eps(\la)}{2 T} \biggr) \biggr]
	      \epp
\end{equation}
Another function needed below is the square of a generalized Cauchy
determinant,
\begin{equation}
     {\cal D} \bigl(\{x_j\}_{j=1}^m, \{y_k\}_{k=1}^n\bigr) =
        \frac{\bigl[ \prod_{1 \le j < k \le m} \sh^2 (x_j - x_k) \bigr]
              \bigl[ \prod_{1 \le j < k \le n} \sh^2 (y_j - y_k) \bigr]}
             {\prod_{j=1}^m \prod_{k=1}^n \sh^2 (x_j - y_k)} \epp
\end{equation}

\begin{figure}
\begin{center}
\includegraphics[width=.92\textwidth]{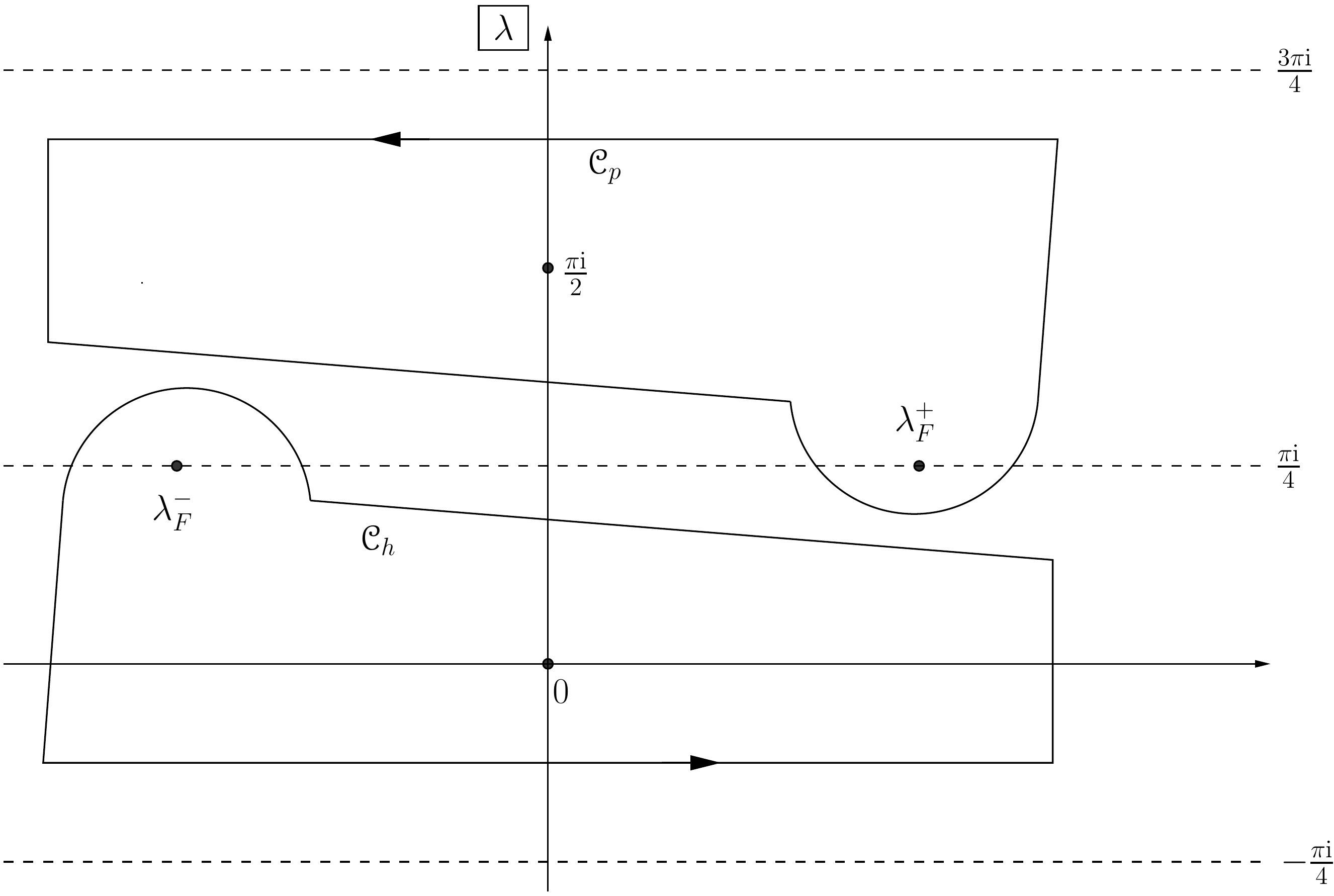}
\caption{\label{fig:xx_ch_and_cp} Sketch of the hole and
particle contours ${\cal C}_h$ and ${\cal C}_p$. The Fermi rapidity $\la_F^-$
is located inside ${\cal C}_h$, while $\la_F^+$ lies inside ${\cal C}_p$.}
\end{center}
\end{figure}

Two more functions that will play an important role in our analysis are
defined as periodic Cauchy transforms with respect to a `hole contour'
${\cal C}_h$ and a `particle contour' ${\cal C}_p$ sketched in
figure~\ref{fig:xx_ch_and_cp}. The same simple, positively oriented contours
will occur in the definition of the form factor series. They are defined in
such a way that ${\cal C}_h$ encloses all roots of $\re^{- \eps(x)/T} - 1$
inside the strip $- \p/4 < \Im x < \p/4$ as well as the left Fermi rapidity
$\la_F^-$, but no other roots of $\re^{- \eps(x)/T} - 1$, while ${\cal C}_p$
encloses the roots of $\re^{- \eps(x)/T} - 1$ inside the strip $\p/4 <
\Im x < 3 \p/4$ as well as the right Fermi rapidity $\la_F^+$ and again no
other roots of $\re^{- \eps(x)/T} - 1$.

Given these contours we define
\begin{equation} \label{phh}
     \PH_h (x) = \frac{\i p'(x)}{2}
                 \exp \biggl\{ \i \int_{{\cal C}_h} \rd \la \: p'(\la) z(\la)
	                       \frac{\sh(x + \la)}{\sh(x - \la)} \biggr\}
\end{equation}
for all $x \in {\cal S} \setminus {\cal C}_h$, and
\begin{equation} \label{php}
     \PH_p (x) = \frac{\i p'(x)}{2}
                 \exp \biggl\{ - \i \int_{{\cal C}_p} \rd \la \: p'(\la) z(\la)
	                       \frac{\sh(x + \la)}{\sh(x - \la)} \biggr\}
\end{equation}
for all $x \in {\cal S} \setminus {\cal C}_p$. Being Cauchy transforms
the functions $\PH_h$ and $\PH_p$ have jump discontinuities across the
contours ${\cal C}_h$ and ${\cal C}_p$, respectively. Hence, each of these
functions determines two functions on the respective contour by its
boundary values from inside and outside the contour. We denote these
functions by $\PH_{h\pm}$ and $\PH_{p\pm}$, where the plus sign stands for
the boundary value from the left and the minus sign for the boundary value
from the right of an oriented contour.

Fix $x \in \Int ({\cal C}_h) \cup \Int ({\cal C}_p)$. Then $\sh(x + \la)/\sh(x - \la)$
is a holomorphic function of $\la$ for all $\la \in {\cal S} \setminus
\bigl(\Int ({\cal C}_h) \cup \Int ({\cal C}_p)\bigr)$. Since the integrands
in (\ref{phh}), (\ref{php}) are rapidly decaying for $\la \rightarrow \pm \infty$,
we may deform the contours and conclude that
\begin{subequations}
\begin{align}
     \PH_h (x) & = \PH_p (x) && \text{for all $x \in \Int ({\cal C}_h) \cup \Int ({\cal C}_p$),} \\
     \PH_{h+} (x) & = \PH_p (x) && \text{for all $x \in {\cal C}_h$,} \\
     \PH_{p+} (x) & = \PH_h (x) && \text{for all $x \in {\cal C}_p$.}
\end{align}
\end{subequations}

\subsection{The series}
We can now recall the form factor series derived in \cite{GKKKS17}. Using the
notation introduced above and performing some rather obvious simplifications
it can be written as
\begin{multline} \label{ffseriesxxtrans}
     \bigl\<\s_1^- \s_{m+1}^+ (t)\bigr\>_T = (-1)^m {\cal F} (m)
        \sum_{n=1}^\infty \frac{(-1)^n}{n! (n-1)!}
	   \prod_{j=1}^n \int_{{\cal C}_h} \frac{\rd x_j}{\p \i}
	      \frac{\PH_p (x_j) \re^{\i (m p(x_j) - t \eps(x_j))}}{1 - \re^{\eps (x_j)/T}}
	      \\ \times
	   \prod_{k=1}^{n-1} \int_{{\cal C}_p} \frac{\rd y_k}{\p \i}
	      \frac{\re^{- \i (m p(y_k) - t \eps(y_k))}}
	           {\PH_h (y_k) \bigl(1 - \re^{- \eps(y_k)/T}\bigr)} \:
		    {\cal D} \bigl( \{x_j\}_{j=1}^n, \{y_k\}_{k=1}^{n-1} \bigr) \epc
\end{multline}
where
\begin{multline}
     {\cal F} (m) =
        \exp \biggl\{- \int_{{\cal C}_h' \subset {\cal C}_h} \rd \la \: z(\la)
	\int_{{\cal C}_h} \rd \m \: \cth' (\la - \m) z(\m) \\
	- \i m p_F - m \int_{{\cal C}_h} \frac{\rd \la}{2 \p} p' (\la)
	\ln \biggl| \cth \biggl(\frac{\eps(\la)}{2 T} \biggr) \biggr| \biggr\}
	\epp
\end{multline}
The contour ${\cal C}_h'$ is tightly enclosed by ${\cal C}_h$.

\section{Fredholm determinant representation} \label{sec:freddetrep}
For the purpose of this work the series on the right hand side of
(\ref{ffseriesxxtrans}) defines the correlation function. The task
is to evaluate it numerically or asymptotically. An important step
towards its efficient evaluation will be to rewrite the series as a
Fredholm series, using a technique developed by Korepin and Slavnov
in \cite{KoSl90}. This is possible due to the occurrence of the
square of the generalized Cauchy determinant on the right hand side
of (\ref{ffseriesxxtrans}). Let
\begin{equation}
      g(x) = \i (mp(x) - t \eps(x))
\end{equation}
and
\begin{equation}
     w(x) = \frac{1}{\p \i \PH_h (x) \bigl(1 - \re^{- \eps(x)/T}\bigr)} \epc \qd
     \overline w(x) = - \frac{\PH_p (x)}{\p \i \bigl(1 - \re^{\eps(x)/T}\bigr)} \epp
\end{equation}
Then
\begin{multline} \label{ffseriessimple}
     \bigl\<\s_1^- \s_{m+1}^+ (t)\bigr\>_T = (-1)^m {\cal F} (m)
        \sum_{n=1}^\infty \frac{1}{n! (n-1)!}
	   \prod_{j=1}^n \int_{{\cal C}_h} \rd x_j \: \overline w (x_j) \re^{g(x_j)}
	      \\ \times
	   \prod_{k=1}^{n-1} \int_{{\cal C}_p} \rd y_k \: w(y_k) \re^{-g(y_k)} \:
		    {\cal D} \bigl( \{x_j\}_{j=1}^n, \{y_k\}_{k=1}^{n-1} \bigr) \epp
\end{multline}

For the summation we note that ${\cal D} = C^2$, where
\begin{equation}
     C = \lim_{\Re y_n \rightarrow + \infty} \2 \det_n \bigr(\ph(x_j,y_k) \bigl) \epc \qd
     \ph(x,y) = \frac{\re^{y - x}}{\sh(y - x)} \epp
\end{equation}
It follows that
\begin{align}
     & \frac{1}{n!}\prod_{j=1}^n \int_{{\cal C}_h} \rd x_j \: \overline w (x_j) \re^{g(x_j)}
		    {\cal D} \bigl( \{x_j\}_{j=1}^n, \{y_k\}_{k=1}^{n-1} \bigr)
		    \notag \\ & \mspace{72.mu}
       = \lim_{\Re y_n \rightarrow + \infty}
       \frac{1}{n!}\prod_{j=1}^n \int_{{\cal C}_h} \rd x_j \: \overline w (x_j) \re^{g(x_j)}
       \notag \\[-1ex] & \mspace{90.mu} \times
       \sum_{P \in \mathfrak{S}_n} \sign(P) \ph(x_1, y_{P1}) \dots \ph(x_n,y_{Pn}) \4
           \det_n \bigr(\ph(x_j,y_k) \bigl)
	   \notag \\ & \mspace{72.mu}
       = \lim_{\Re y_n \rightarrow + \infty}
         \4 \det_n \biggl( \int_{{\cal C}_h} \rd x \: \overline w(x) \re^{g(x)} \ph(x,y_j) \ph(x,y_k) \biggr)
	 \epp
\end{align}
Setting
\begin{subequations}
\begin{align}
     \widetilde V(y_j,y_k) & = 
         \int_{{\cal C}_h} \rd x \: \overline w(x) \re^{g(x)} \ph(x,y_j) \ph(x,y_k)
	 \epc \label{vtilde} \\
     E_h (y_j) & = \int_{{\cal C}_h} \rd x \: \overline w(x) \re^{g(x)} \ph(x,y_j) \epc \\
     \Om & = \int_{{\cal C}_h} \rd x \: \overline w(x) \re^{g(x)} \epc
\end{align}
\end{subequations}
taking into account that $\lim_{\Re y \rightarrow + \infty} \ph(x,y) = 2$ and
applying elementary manipulations to the determinant we obtain
\begin{multline} \label{prefreddet}
     \frac{1}{n!}\prod_{j=1}^n \int_{{\cal C}_h} \rd x_j \: \overline w (x_j) \re^{g(x_j)}
        {\cal D} \bigl( \{x_j\}_{j=1}^n, \{y_k\}_{k=1}^{n-1} \bigr) \\[-1ex]
	= \Om \, \det_{n-1} \Bigl( \widetilde V(y_j,y_k) - E_h(y_j) E_h(y_k)/\Om \Bigr) \epp
\end{multline}
Inserting this expression into the right hand side of (\ref{ffseriessimple})
we end up with a Fredholm series which can be interpreted as
\begin{equation} \label{freddetwithp}
     \bigl\<\s_1^- \s_{m+1}^+ (t)\bigr\>_T =
        (-1)^m {\cal F} (m) \Om \, \det_{{\cal C}_p} \bigl(\id + \widehat V - \widehat P\bigr) \epc
\end{equation}
where the integral operators $\widehat V$ and $\widehat P$ are defined
relative to the contour ${\cal C}_p$,
\begin{subequations}
\begin{align}
     \widehat V f (x) & = \int_{{\cal C}_p} \rd y \: w(y) \re^{- g(y)} \widetilde V(x,y) f(y) \epc \\[1ex]
     \widehat P f (x) & = \frac{E_h(x)}{\Om} \int_{{\cal C}_p} \rd y \: w(y) \re^{- g(y)} E_h (y) f(y) \epp
\end{align}
\end{subequations}

Introducing the resolvent $\widehat R$ of $\widehat V$ in the usual way, such that
\begin{equation}
     (\id - \widehat R) (\id + \widehat V) = (\id + \widehat V) (\id - \widehat R) = \id
\end{equation}
we can rewrite the Fredholm determinant in (\ref{freddetwithp}) and finally
obtain the following
\begin{theorem} \label{th:freddetrep}
The transverse correlation function of the XX chain admits the Fredholm
determinant representation
\begin{multline} \label{freddetrep}
     \<\s_1^- \s_{m+1}^+ (t)\>_T = (-1)^m {\cal F} (m) \\
        \times \biggl[ \Om - \int_{{\cal C}_p} \rd x \: w(x) \re^{- g(x)}
	       E_h (x) (\id - \widehat R) E_h (x) \biggr]
	       \det_{{\cal C}_p} \bigl(\id + \widehat V \bigr) \epp
\end{multline}
\end{theorem}
\section{\boldmath $\widehat V$ as an integrable integral operator}
\label{sec:intop}
We observe that
\begin{equation}
     \ph(x,y_j) \ph(x,y_k) = - \frac{\re^{- y_j - y_k}}{\sh(y_j - y_k)}
                               \bigl[\re^{2 y_k} \ph(x,y_j) - \re^{2 y_j} \ph(x,y_k)\bigr] \epp
\end{equation}
Inserting this identity into (\ref{vtilde}) we obtain the following
expression for the kernel $V$ of the integral operator $\widehat V$,
\begin{multline}
     V(y_j,y_k) = \widetilde V(y_j,y_k) w(y_k) \re^{- g(y_k)} \\
        = \frac{\re^{- y_j - y_k}}{\sh(y_j - y_k)}
	  \bigl[ \re^{2 y_j} E_h(y_k) - \re^{2 y_k} E_h(y_j) \bigr] w(y_k) \re^{- g(y_k)} \epp
\end{multline}

Define
\begin{equation} \label{defeler}
     \Ev_L (x) = \begin{pmatrix} \re^{2x} \\ - E_h (x) \end{pmatrix} \epc \qd
     \Ev_R (x) = \begin{pmatrix} E_h (x) \\ \re^{2x} \end{pmatrix} w(x) \re^{-g(x)} \epp
\end{equation}
Then $V(x,y)$ can be written as
\begin{equation} \label{intkernel}
     V(x,y) = \frac{\re^{- x -y}}{\sh(x - y)} \Ev_L^t (x) \Ev_R (y) \epp
\end{equation}
The vectors $\Ev_L$ and $\Ev_R$ have the important property that
\begin{equation} \label{orthokernel}
     \Ev_L^t (x) \Ev_R (x) = 0 \epp
\end{equation}
Integral operators with a kernel satisfying (\ref{intkernel}), (\ref{orthokernel})
are rather special. They fall into the class of integrable integral operators
\cite{IIKS90,Deift99,Sakhnovich68}. Having in mind a possible extension of
our work to the more general XXZ chain, we have employed a parameterisation
in terms of rapidity variables. For this reason the integrable kernel
(\ref{intkernel}) contains hyperbolic rather than rational functions and
looks slightly different from what is commonly encountered in the literature.

As observed in \cite{IIKS90} an important property of an integrable
integral operator is that its resolvent kernel is as well of the form
(\ref{intkernel}). The resolvent kernel $R(x,y)$ is defined as the
solution of the linear integral equation
\begin{equation} \label{resolventkernel}
     R(x,y) = V(x,y) - \int_{{\cal C}_p} \rd z \: V(x,z) R(z,y) \epp
\end{equation}
Let
\begin{equation}
     f_y (x) = \2 (\re^{2x} - \re^{2y}) R(x,y) \epp
\end{equation}
Inserting (\ref{intkernel}) into (\ref{resolventkernel}) and multiplying by
$\re^{x + y} \sh(x - y) = \2 (\re^{2x} - \re^{2y})$ we see that
\begin{equation}
     \Ev_L^t (x) \Bigl( \Ev_R (x) - \int_{{\cal C}_p} \rd z \: \Ev_R (z) R(z,y) \Bigr)
        = \bigl(\id + \widehat V\bigr) f_y (x) \epp
\end{equation}
Upon acting with $\id - \widehat R$ on this equation we conclude that
\begin{equation} \label{resolventform}
     R(x,y) = \frac{\re^{- x -y}}{\sh(x - y)} \Fv_L^t (x) \Fv_R (y) \epc
\end{equation}
where
\begin{subequations}
\begin{align} \label{resolveflfr}
     & \Fv_L^t (x) = \Ev_L^t (x) - \int_{{\cal C}_p} \rd z \: R(x,z) \Ev_L^t (z) \epc \\
     & \Fv_R (y) = \Ev_R (y) - \int_{{\cal C}_p} \rd z \: \Ev_R (z) R(z,y) \epp
\end{align}
\end{subequations}
The latter pair of equations is equivalent to saying the $\Fv_L^t$ and
$\Fv_R$ satisfy the linear integral equations
\begin{subequations}
\label{intsf}
\begin{align} \label{lieflt}
     & \Fv_L^t (x) = \Ev_L^t (x) - \int_{{\cal C}_p} \rd z \: V(x,z) \Fv_L^t (z) \epc \\
     & \Fv_R (x) = \Ev_R (x) - \int_{{\cal C}_p} \rd z \: \Fv_R (z) V(z,x) \epp \label{liefr}
\end{align}
\end{subequations}
Note that we did not use (\ref{orthokernel}), when we derived (\ref{resolventform}),
(\ref{intsf}). Further note that we assumed that $\id + \widehat V$ is invertible.
This implies that $\Fv_L^t$ and $\Fv_R$ are uniquely determined and that $\det_{{\cal C}_p}
(\id + \widehat V)$ is non-zero.

\section{The matrix Riemann-Hilbert problem} \label{sec:mRHp}
Define a $2 \times 2$ matrix function
\begin{equation} \label{defchi}
     \chi (x) = I_2 + \int_{{\cal C}_p} \rd y \:
                \frac{\re^{- x - y}}{\sh(x - y)} \Fv_R (y) \Ev_L^t (y)
\end{equation}
on ${\mathbb C} \setminus \bigl({\cal C}_p \mod \p \i\bigr)$. This function
connects $\Ev_R$ and $\Fv_R$ algebraically,
\begin{equation} \label{erfralg}
     \chi_+ (x) \Ev_R (x) = \Fv_R (x)
\end{equation}
for all $x \in {\cal C}_p$, which follows from (\ref{liefr}).

Following \cite{IIKS90} we shall now establish, for the reader's
convenience, several properties of $\chi$, using the hyperbolic form
of the kernel.
\begin{proposition} \label{prop:mrhpprops}
The following properties of $\chi$ follow immediately from its definition
and from (\ref{erfralg}).
\begin{enumerate}
\item
$\chi$ is $\p \i$-periodic and holomorphic in ${\mathbb C} \setminus
\bigl({\cal C}_p \mod \p \i\bigr)$.
\item
Asymptotically for large argument in ${\cal S}$ the functions $\chi$
behaves as
\begin{equation} \label{asychi}
     \chi(x) \sim \begin{cases}
                     I_2 - 2 \int_{{\cal C}_p} \rd y \: \re^{- 2y} \Fv_R (y) \Ev_L^t (y)
		     + {\cal O} \bigl(\re^{2x}\bigr) &
		     \text{for $\Re x \rightarrow - \infty$,} \\[1ex]
                     I_2 + {\cal O} \bigl(\re^{-2x}\bigr) &
		     \text{for $\Re x \rightarrow + \infty$,}
                  \end{cases}
\end{equation}
where it is understood that the symbols ${\cal O} \bigl(\re^{\pm 2x}\bigr)$
refer to the behaviour of the individual matrix elements.
\item
The function $\chi$ admits continuous boundary values from in- and
outside ${\cal C}_p$. The corresponding two continuous boundary
functions $\chi_+$ and $\chi_-$ are connected by the multiplicative
jump condition
\begin{equation}
     \chi_- (x) = \chi_+ (x) G_\chi (x) \epc
\end{equation}
where
\begin{equation} \label{defgchi}
     G_\chi (x) = I_2 + 2 \p \i \re^{-2x} \Ev_R (x) \Ev_L^t (x) \epp
\end{equation}
\end{enumerate}
\end{proposition}

Another immediate consequence of the definition is the following
\begin{proposition}
The function $\chi (x)$ defined in (\ref{defchi}) is invertible in
${\cal S} \setminus {\cal C}_p$ with inverse
\begin{equation} \label{chiinv}
     \chi^{-1} (x) = I_2 - \int_{{\cal C}_p} \rd y \:
                \frac{\re^{- x - y}}{\sh(x - y)} \Ev_R (y) \Fv_L^t (y) \epp
\end{equation}
\end{proposition}
The proof is by direct calculation, using the linear integral equations
(\ref{intsf}). It follows from (\ref{chiinv}) that
\begin{equation} \label{elflalg}
     \Ev_L^t (x) \chi_+^{-1} (x) = \Fv_L^t (x) \epp
\end{equation}
Combining (\ref{erfralg}) and (\ref{elflalg}) with (\ref{orthokernel})
we conclude that
\begin{equation} \label{orthorkernel}
     \Fv_L^t (x) \Fv_R (x) = 0 \epp
\end{equation}
Taking into account (\ref{resolventform}) this means that the resolvent of
an integrable integral operator satisfying (\ref{orthokernel}) is an
integrable integral operator of the same type.

In proposition~\ref{prop:mrhpprops} we have shown that the matrix
function $\chi$ defined in (\ref{defchi}) has certain properties.
One may reverse the problem and ask if a matrix function with
these properties exists and is unique for a given contour and
a given jump matrix $G_\chi$.
\begin{definition}
We shall say that a holomorphic function $\Ps: {\cal S} \setminus
{\cal C}_p \mapsto {\rm Mat}_{2 \times 2} \bigl( {\mathbb C} \bigr)$ 
solves the matrix Riemann-Hilbert problem $({\cal S}, {\cal C}_p,
+ \infty, G_\Ps)$ with jump matrix $G_\Ps: {\cal C}_p \mapsto
{\rm Mat}_{2 \times 2} \bigl( {\mathbb C} \bigr)$ on a cylinder
$\cal S$ if
\begin{enumerate}
\item
\begin{equation}
     \lim_{\Re x \rightarrow \pm \infty} \Ps (x)
        = \Ps^{(\pm)} \bigl(I_2 + {\cal O} \bigl( \re^{\mp 2x} \bigr)\bigr)
\end{equation}
where $\Ps^{(+)} = I_2$ and $\Ps^{(-)}$ is a constant matrix,
\item
$\Ps$ has continuous boundary values on ${\cal C}_p$ from inside
and outside and satisfies the multiplicative jump condition
\begin{equation}
      \Ps_- (x) = \Ps_+ (x) G_\Ps (x)
\end{equation}
on ${\cal C}_p$.
\end{enumerate}
\end{definition}
Thus, $\chi$ defined in (\ref{defchi}) solves the matrix
Riemann-Hilbert problem $({\cal S}, {\cal C}_p, + \infty, G_\chi)$ with
jump matrix $G_\chi$ defined in (\ref{defgchi}). In the following we
prove that the solution is unique. This fact will turn the matrix
Riemann-Hilbert problem into a useful tool, since, if we are able to
find a solution by any means, this solution then determines the
resolvent of $\widehat V$ through (\ref{resolventform}), (\ref{erfralg}),
(\ref{elflalg}).

\begin{lemma} \label{lem:invchiex}
A solution $\Ps$ of the matrix Riemann-Hilbert problem $({\cal S},
{\cal C}_p, + \infty, G_\Ps)$ is invertible if $G_\Ps$ is unimodular
($\det G_\Ps = 1$).
\end{lemma}
\begin{proof}
Let $d(x) = \det \Ps (x)$. Then $d(x + \p \i) = d(x)$ and $d$ is
holomorphic in ${\cal S} \setminus {\cal C}_p$. The function $d$
behaves asymptotically as
\begin{equation} \label{dasy}
     d(x) \sim \begin{cases}
		  \text{const.} + {\cal O} \bigl(\re^{2x}\bigr) &
		  \text{for $\Re x \rightarrow - \infty$,} \\[1ex]
		  1 + {\cal O} \bigl(\re^{-2x}\bigr) &
		  \text{for $\Re x \rightarrow + \infty$.}
               \end{cases}
\end{equation}
Moreover, $d$ satisfies the jump condition
\begin{equation} \label{djumpnojump}
     d_- (x) = d_+ (x) \epc
\end{equation}
which follows from $\det G_\Ps (x) = 1$ for all $x \in {\cal C}_p$.

Choose a simple closed contour ${\cal C}$ such that ${\cal C}_p 
\subset \Int {\cal C} \subset {\cal S}$ and define
\begin{equation}
     I(x) = \int_{\cal C} \frac{\rd y}{2\p\i} \: \ph(x,y) d(y) \epc
\end{equation}
where $x \in \Int {\cal C} \setminus {\cal C}_p$. We calculate this integral
in two different ways, first by making the contour ${\cal C}$ larger and
using the $\p \i$-periodicity and the asymptotic behaviour (\ref{dasy}) of $d$,
and second by shrinking the contour and using the holomorphicity of $d$ and
the `jump condition' (\ref{djumpnojump}). Then, on the one hand, there is
$R_0 > 0$ such that for all $R > R_0$
\begin{equation}
     I(x) = \biggl( \int_{- R + 3\p\i/4}^{- R - \p\i/4} + \int_{R - \p\i/4}^{R + 3\p\i/4} \biggr)
            \frac{\rd y}{2\p\i} \ph(x,y) d(y) \epp
\end{equation}
Sending $R$ to $+\infty$ and using (\ref{dasy}) we conclude that $I(x) = 1$
for all $x \in {\cal S} \setminus {\cal C}_p$.

On the other hand, if e.g.\ $x$ outside ${\cal C}_p$, then
\begin{multline}
     I(x) = d(x) + \int_{{\cal C}_p} \frac{\rd y}{2\p\i} \: \ph(x,y) d_-(y)
          = d(x) + \int_{{\cal C}_p} \frac{\rd y}{2\p\i} \: \ph(x,y) d_+(y) \\
          = d(x) + \int_{{\cal C}'} \frac{\rd y}{2\p\i} \: \ph(x,y) d (y) = d(x) \epp
\end{multline}
where ${\cal C}' \subset \Int {\cal C}_p$ is a simple closed contour. Here
we have used the residue theorem in the first equation, (\ref{djumpnojump})
in the second equation and Cauchy's theorem in the third and fourth equation.
If $x$ is inside ${\cal C}_p$, the same arguments can be employed in a
different order. Altogether we have shown that $d(x) = 1$ for all $x \in
{\cal S} \setminus {\cal C}_p$, which proves the lemma.
\end{proof}

Based on lemma~\ref{lem:invchiex} we can now prove the following uniqueness
theorem.
\begin{theorem} \label{th:mrhpunique}
If $G_\Ps$ is unimodular, the matrix Riemann-Hilbert problem
$({\cal S}, {\cal C}_p, + \infty, G_\Ps)$ admits at most a single solution.
\end{theorem}
\begin{proof}
Consider two solutions $\Ps_1$, $\Ps_2$ of the matrix Riemann-Hilbert problem.
Then
\begin{equation}
     \Ps_2 (x) = \begin{pmatrix}
                     \Ps_2^{11} (x) & \Ps_2^{12} (x) \\
                     \Ps_2^{21} (x) & \Ps_2^{22} (x)
		  \end{pmatrix}
\end{equation}
is invertible for every $x \in {\cal S} \setminus {\cal C}_p$ by
lemma~\ref{lem:invchiex}. And
\begin{equation} \label{chi2inv}
     \Ps_2^{-1} (x) = \begin{pmatrix}
                          \Ps_2^{22} (x) & - \Ps_2^{12} (x) \\
			  - \Ps_2^{21} (x) & \Ps_2^{11} (x)
		       \end{pmatrix} \epc
\end{equation}
since $\det \Ps_2 (x) = 1$ according to the same lemma. Equation
(\ref{chi2inv}) implies that $\Ps_2^{-1}$ inherits the characteristic
properties of $\Ps_2$. $\Ps_2^{-1}$ is $\p\i$-periodic, holomorphic
in ${\cal S} \setminus {\cal C}_p$ and behaves asymptotically for
large spectral parameter inside ${\cal S}$ as
\begin{equation} \label{asychi2inv}
     \Ps_2^{-1} (x) \sim
                     \begin{cases}
		        \bigl(\Ps_2^{(-)}\bigr)^{-1} + {\cal O} \bigl(\re^{2x}\bigr) &
			\text{for $\Re x \rightarrow - \infty$,} \\[1ex]
			I_2 + {\cal O} \bigl(\re^{-2x}\bigr) &
			\text{for $\Re x \rightarrow + \infty$.}
                     \end{cases}
\end{equation}
Since $G_\Ps (x)$ is invertible for all $x \in {\cal C}_p$
we may invert the jump condition for $\Ps_2$ to obtain
\begin{equation} \label{inversejump}
     \Ps_{2-}^{-1} (x) = G_\Ps^{-1} (x) \Ps_{2+}^{-1} (x)
\end{equation}
for all $x \in {\cal C}_p$.

Now let
\begin{equation}
     W(x) = \Ps_1 (x) \Ps_2^{-1} (x) \epp
\end{equation}
Clearly $W$ is $\p\i$-periodic and holomorphic in ${\cal S} \setminus {\cal C}_p$.
Its asymptotic behaviour can be inferred from (\ref{asychi2inv}) and is again
of the same form,
\begin{equation} \label{asyw}
     W (x) \sim \begin{cases}
                   W^{(-)} + {\cal O} \bigl(\re^{2x}\bigr) &
		   \text{for $\Re x \rightarrow - \infty$,} \\[1ex]
		   I_2 + {\cal O} \bigl(\re^{-2x}\bigr) &
		   \text{for $\Re x \rightarrow + \infty$,}
                \end{cases}
\end{equation}
where $W^{(-)}$ is a constant matrix. Moreover, (\ref{inversejump}) combined with
the jump condition for $\Ps_1$ implies that
\begin{equation}
    W_- (x) = W_+ (x) \epc
\end{equation}
and we can proceed as in the proof of the previous lemma.
We fix a simple closed contour ${\cal C}$ such that ${\cal C}_p 
\subset \Int {\cal C} \subset {\cal S}$ and define
\begin{equation}
     J(x) = \int_{\cal C} \frac{\rd y}{2\p\i} \: \ph(x,y) W(y) \epc
\end{equation}
where $x \in \Int {\cal C} \setminus {\cal C}_p$. Calculating this integral
again in two different ways, exactly as in the proof of lemma~\ref{lem:invchiex},
we reach the conclusion that $W(x) = I_2$ for all $x \in \Int {\cal C} \setminus
{\cal C}_p$ which completes our proof.
\end{proof}
It follows that the matrix Riemann-Hilbert problem $({\cal S}, {\cal C}_p,
+ \infty, G_\chi)$ with $G_\chi$ defined in (\ref{defgchi}) has a
unique solution, since
\begin{equation}
     \det G_\chi (x) = 1 + 2 \p \i \re^{-2x} \Ev_L^t (x) \Ev_R (x) = 1
\end{equation}
for all $x \in {\cal C}_p$. Here we have used (\ref{orthokernel}) in the
last equation.

\section{Expressing the correlation function in terms of the matrix
Riemann-Hilbert problem} \label{sec:corrmRHp}
Our goal is to perform a direct asymptotic analysis of the matrix
Riemann-Hilbert problem and to use the result in order to determine
the asymptotics of the transversal correlation functions of the XX
chain by means of theorem~\ref{th:freddetrep}. For this purpose we
follow \cite{KKMST09c,Kozlowski11b} and express the resolvent dependent
factor and the Fredholm determinant factor on the right hand side of
(\ref{freddetrep}) in terms of the solution $\chi$ of our matrix
Riemann-Hilbert problem.
\begin{proposition} \label{prop:resolvepart}
The resolvent dependent factor in the Fredholm determinant representation
(\ref{freddetrep}) of the transversal correlation function of the XX chain
is connected to the solution of the matrix Riemann-Hilbert problem defined
in proposition~\ref{prop:mrhpprops} by the identity
\begin{equation} \label{respartmrhp}
     \int_{{\cal C}_p} \rd x \: w(x) \re^{- g(x)} E_h (x) (\id - \widehat R) E_h (x)
        = - \2 \lim_{\Re x \rightarrow + \infty} \re^{2x} \chi^{12} (x) \epp
\end{equation}
\end{proposition}
\begin{proof}
It follows from (\ref{defeler}) and (\ref{resolveflfr}) that
\begin{equation} \label{respartef}
     \int_{{\cal C}_p} \rd x \: w(x) \re^{- g(x)} E_h (x) (\id - \widehat R) E_h (x)
        = - \int_{{\cal C}_p} \rd x \: \Ev_R^{1} (x) \Fv_L^{t \, 2} (x) \epp
\end{equation}
On the other hand,
\begin{equation} \label{chi12ef}
     \chi^{12} (x) = - \bigl(\chi^{-1}\bigr)^{12} =
        \re^{-2x} \int_{{\cal C}_p} \rd y \: \ph(y,x) \Ev_R^{1} (y) \Fv_L^{t \, 2} (y)
\end{equation}
which follows from the fact that $\det \chi(x) = 1$ and from equation
(\ref{chiinv}). Comparing (\ref{respartef}) and (\ref{chi12ef}) for large
positive real part of the argument we obtain (\ref{respartmrhp}).
\end{proof}

As for the Fredholm determinant on the right hand side of (\ref{freddetrep})
we shall derive a formula for the logarithmic time derivative of
$\det_{{\cal C}_p} \bigl(\id + \widehat V\bigr)$. For this purpose we
shall utilize the identity
\begin{equation} \label{dlogdetisdtrlog}
     \6_t \ln\Bigl(\det_{{\cal C}_p} \bigl(\id + \widehat V\bigr)\Bigr) =
        \tr \bigl\{ \dot{\widehat V} \bigl(\id - \widehat R\bigr)\bigr\} \epp
\end{equation}
Here and in the following the dot stands for the time derivative.

\begin{proposition} \label{prop:dtlnfred}
The Fredholm determinant in (\ref{freddetrep}) is related to the solution
of the matrix Riemann-Hilbert problem formulated in
proposition~\ref{prop:mrhpprops} by the identity
\begin{equation} \label{dtlnfredmrhp}
     \6_t \ln\Bigl(\det_{{\cal C}_p} \bigl(\id + \widehat V\bigr)\Bigr) =
        \int_{\g ({\cal C}_p)} \frac{\rd z}{2\p} \: \eps(z)
	\tr \bigl\{ \chi'(z) \begin{pmatrix} 0 & \re^{-2z} E_p (z) \\ 0 & 1 \end{pmatrix}
	            \chi(z) \bigr\} \epp
\end{equation}
Here $\g ({\cal C}_p)$ is a contour surrounding ${\cal C}_p$ closely enough,
and the function $E_p$ is defined by
\begin{equation} \label{defep}
      E_p (x) =  \int_{{\cal C}_p} \frac{\rd y}{\p\i} \:
                 \frac{\PH_{p-} (y) \re^{g(y)}}{1 - \re^{\eps(y)/T}} \ph(y,x)
\end{equation}
for all $x \in {\cal S} \setminus {\cal C}_p$.
\end{proposition}
Before proceeding with the proof let us note the following relations between
$E_p$ defined in the proposition and the function $E_h$ defined earlier,
\begin{subequations}
\begin{align} \label{ehep}
     E_h (x) & = E_p (x) && \text{for all $x \in \Int ({\cal C}_h) \cup \Int ({\cal C}_p$),} \\
     E_{h+} (x) & = E_p (x) && \text{for all $x \in {\cal C}_h$,} \label{ehpep} \\
     E_{p+} (x) & = E_h (x) && \text{for all $x \in {\cal C}_p$.} \label{eppeh}
\end{align}
\end{subequations}
Clearly (\ref{ehpep}) and (\ref{eppeh}) are consequences of (\ref{ehep}), which
can be obtained by deforming the contour in the definition of $E_h$ and
using the analytic and asymptotic properties of the integrand.
\begin{proof}
The first step of the proof consists in deriving an appropriate representation
for $\dot V (x,y)$. First of all
\begin{multline} \label{vdot}
     \dot V(x,y) = - \frac{\re^{- x - y}}{\sh(x - y)}
                     \biggl[\re^{2y} \bigl(\dot E_h (x) - \dot g(x) E_h (x) \bigr)
                          - \re^{2x} \bigl(\dot E_h (y) - \dot g(y) E_h (y) \bigr) \\
			  + \bigl(\dot g(x) - \dot g(y) \bigr) \re^{2y} E_h (x) \biggr]
			    w(y) \re^{- g(y)} \epp
\end{multline}

Let $x \in {\cal C}_p$. Then
\begin{multline}
     - \int_{\g ({\cal C}_p)} \frac{\rd z}{2\p\i} \:
          \frac{\re^{- z - x}}{\sh(z - x)} \dot g(z) E_p (z)
	  = \re^{-2x} \bigl(\dot E_{p+} (x) - \dot g(x) E_{p+} (x) \bigr) \\
	  = \re^{-2x} \bigl(\dot E_h (x) - \dot g(x) E_h (x) \bigr) \epc
\end{multline}
where we have used (\ref{defep}) in the first equation and (\ref{eppeh})
in the second equation. We also observe that
\begin{equation}
     \re^{2x + 2y} w(y) \re^{- g(y)} = \Ev_L^t (x) \s^+ \Ev_R (y) \epp
\end{equation}
Combining the latter two equations we obtain
\begin{multline} \label{vdotfirsthalf}
     - \frac{\re^{- x - y}}{\sh(x - y)}
       \biggl[\re^{2y} \bigl(\dot E_h (x) - \dot g(x) E_h (x) \bigr)
            - \re^{2x} \bigl(\dot E_h (y) - \dot g(y) E_h (y) \bigr) \biggr] w(y) \re^{- g(y)} \\
     = \int_{\g ({\cal C}_p)} \frac{\rd z}{2\p\i} \:
          \frac{\re^{- x - y}}{\sh(z - x)\sh(z - y)}
	  \dot g(z) \re^{- 2z} E_p (z) \Ev_L^t (x) \s^+ \Ev_R (y) \epp
\end{multline}
For the remaining terms in (\ref{vdot}) we use that
\begin{equation}
     \int_{\g ({\cal C}_p)} \frac{\rd z}{2\p\i} \:
          \frac{\dot g(z)}{\sh(z - x) \sh(z - y)} = \frac{\dot g(x) - \dot g(y)}{\sh(x - y)}
\end{equation}
for all $x, y \in {\cal C}_p$, $x \ne y$ and that
\begin{equation}
     - \re^{2y} E_h (x) w(y) \re^{- g(y)} = \Ev_L^t (x) e_2^2 \, \Ev_R (y) \epc
\end{equation}
where $e_2^2$ is one of the matrix units defined by $(e_\a^\be)^\g_\de =
\de^\g_\a \de^\be_\de$. Then
\begin{multline}
     - \frac{\re^{- x - y}}{\sh(x - y)}
       \bigl(\dot g(x) - \dot g(y) \bigr) \re^{2y} E_h (x) w(y) \re^{- g(y)} \\
     = \int_{\g ({\cal C}_p)} \frac{\rd z}{2\p\i} \:
       \frac{\re^{- x - y}}{\sh(z - x) \sh(z - y)} \dot g(z) \Ev_L^t (x) e_2^2 \, \Ev_R (y) \epp
\end{multline}
Inserting the latter identity together with (\ref{vdotfirsthalf})
into (\ref{vdot}) we obtain
\begin{equation} \label{vdotfinal}
     \dot V(x,y) = \int_{\g ({\cal C}_p)} \rd z \:
       \frac{\re^{- x - y}}{\sh(z - x) \sh(z - y)} \Ev_L^t (x)  S(z) \Ev_R (y) \epc
\end{equation}
where
\begin{equation}
     S(z) = \frac{\dot g (z)}{2\p\i}
            \begin{pmatrix}
	       0 & \re^{-2z} E_p (z) \\ 0 & 1
            \end{pmatrix} \epp
\end{equation}

The second step of the proof consists in inserting (\ref{vdotfinal}) into
(\ref{dlogdetisdtrlog}) and does not depend on the precise form of $S(z)$.
\begin{align}
     & \6_t \ln\Bigl(\det_{{\cal C}_p} \bigl(\id + \widehat V\bigr)\Bigr)
        - \int_{{\cal C}_p} \rd x \: \dot V(x,x) =
	- \int_{{\cal C}_p} \rd x \int_{{\cal C}_p} \rd y \: \dot V (x,y) R(y,x) \notag \\[1ex]
	& \mspace{36.mu} 
	= - \int_{\g ({\cal C}_p)} \rd z \int_{{\cal C}_p} \rd x \int_{{\cal C}_p} \rd y \:
	    \frac{8 \re^{2z} \Ev_L^t (x)  S(z) \Ev_R (y) \Fv_L^t (y) \Fv_R (x)}
	         {(\re^{2z} - \re^{2x})(\re^{2z} - \re^{2y})(\re^{2y} - \re^{2x})} \notag \\[2ex]
	& \mspace{36.mu} 
	= - \tr \biggl\{\int_{\g ({\cal C}_p)} \rd z \int_{{\cal C}_p} \rd y \int_{{\cal C}_p} \rd x \: 
	      \biggl[\frac{\re^{- x - z}}{\sh(x - z)} - \frac{\re^{- x - y}}{\sh(x - y)} \biggr]
	      \Fv_R (x) \Ev_L^t (x) \notag \\
	& \mspace{306.mu} \times
	  \frac{\re^{-2y}}{\sh^2(y - z)} S(z) \Ev_R (y) \Fv_L^t (y) \biggr\} \notag \\[1ex]
	& \mspace{36.mu} 
	= - \int_{\g ({\cal C}_p)} \rd z \int_{{\cal C}_p} \rd y \: 
	    \frac{\re^{-2y}}{\sh^2(y - z)}
	    \tr \bigl\{ \bigl(\chi(y) - \chi(z)\bigr) S(z) \Ev_R (y) \Fv_L^t (y) \bigr\} \notag \\[1ex]
	& \mspace{36.mu} 
	= - \int_{{\cal C}_p} \rd y \int_{\g ({\cal C}_p)} \rd z \: \frac{\re^{-2y}}{\sh^2(y - z)}
	    \Ev_L^t (y) S(z) \Ev_R (y) \notag \\
	& \mspace{126.mu} 
	  + \tr \biggl\{\int_{\g ({\cal C}_p)} \rd z \: \chi(z) S(z) \, \6_z
	                \int_{{\cal C}_p} \rd y \: \frac{\re^{- y -z}}{\sh(y - z)}
			\Ev_R (y) \Fv_L^t (y) \biggr\} \epp
\end{align}
Here we have used (\ref{dlogdetisdtrlog}) in the first equation.
In the second equation we inserted (\ref{resolventform}) and
(\ref{vdotfinal}). In the third equation we performed a partial
fraction decomposition and used the fact that $\tr\{\Fv_R (x) \Fv_L^t (y)\}
= \Fv_L^t (y) \Fv_R (x)$. In the fourth equation (\ref{defchi}) was inserted.
In the fifth equation we used (\ref{elflalg}) in the first term on
the right hand side.

Using once more (\ref{vdotfinal}) and equation (\ref{chiinv}) we see
that
\begin{multline}
     \6_t \ln\Bigl(\det_{{\cal C}_p} \bigl(\id + \widehat V\bigr)\Bigr) = \\
          \int_{\g ({\cal C}_p)} \rd z \: \tr \bigr\{\chi(z) S(z) \6_z \chi^{-1} (z) \bigr\}
	= - \int_{\g ({\cal C}_p)} \rd z \: \tr \bigr\{\chi' (z) S(z) \chi^{-1} (z) \bigr\} \\
        = \int_{\g ({\cal C}_p)} \frac{\rd z}{2\p} \: \eps(z)
	  \tr \bigl\{ \chi'(z) \begin{pmatrix} 0 & \re^{-2z} E_p (z) \\ 0 & 1 \end{pmatrix}
	              \chi(z) \bigr\} \epp
\end{multline}
In the last equation we have inserted the explicit forms of $S$ and $g$.
\end{proof}

\section{First transformation of the matrix Riemann-Hilbert problem}
\label{sec:standardform}
Following \cite{IIKV92} we shall transform the matrix
Riemann-Hilbert problem $({\cal S}, {\cal C}_p, + \infty, G_\chi)$
to a certain `normal form' with a simplified jump matrix having no
Cauchy transforms in its entries. For a while we shall use the
notation $f(x) = 2\p\i w(x) \re^{- g(x)}$. Then the jump matrix associated
with the matrix Riemann-Hilbert problem satisfied by $\chi$ defined
in (\ref{defchi}) takes the form
\begin{equation}
     \mspace{-3.mu}
     G_\chi (x) = I_2 + 2 \p \i \re^{-2x} \Ev_R (x) \Ev_L^t (x)
        = \begin{pmatrix}
	     1 + f(x) E_{p+} (x) & - f(x) \re^{-2x} E_{p+}^2 (x) \\
	     f(x) \re^{2x} & 1 - f(x) E_{p+} (x) 
          \end{pmatrix}.
\end{equation}
Here we also used (\ref{eppeh}).

Now define the $2 \times 2$ matrix function
\begin{equation} \label{defchitilde}
     \widetilde \chi (x) = \chi (x)
                           \begin{pmatrix}
			        1 & \re^{-2x} E_p (x) \\ 0 & 1
                           \end{pmatrix} \epp
\end{equation}
This transformation essentially does not change the analytic properties
or the asymptotic behaviour, but the jump matrix of the transformed
function gets modified.
\begin{proposition}
The transformed matrix $\widetilde \chi$ is the unique solution of the
matrix Riemann-Hilbert problem $({\cal S}, {\cal C}_p, + \infty,
G_{\widetilde \chi})$ with jump matrix
\begin{subequations}
\begin{align} \label{transg}
     G_{\widetilde \chi} (x) & =
        \begin{pmatrix}
	   1 & \dst{\frac{2 \PH_{p-} (x) \re^{g(x) - 2x}}{1 - \re^{\eps(x)/T}}} \\[2ex]
	   \dst{\frac{2 \re^{2x - g(x)}}{\PH_{p+} (x) (1 - \re^{- \eps(x)/T})}} &
	   \dst{\frac{\PH_{p-} (x)}{\PH_{p+} (x)}}
	\end{pmatrix} \\[1ex] & =
        \begin{pmatrix}
	   1 & 0 \\[1ex]
	   \dst{\frac{2 \re^{2x - g(x)}}{\PH_{p+} (x) (1 - \re^{- \eps(x)/T})}} & 1
	\end{pmatrix}
        \begin{pmatrix}
	   1 & \dst{\frac{2 \PH_{p-} (x) \re^{g(x) - 2x}}{1 - \re^{\eps(x)/T}}} \\[2ex]
	   0 & 1
	\end{pmatrix} \epp
\end{align}
\end{subequations}
\end{proposition}
\begin{proof}
Analytic properties and $\p\i$-periodicity are clear. The asymptotic behaviour
follows from
\begin{equation}
     \re^{-2x} E_p (x) = \begin{cases}
			    \text{const.} + {\cal O} (\re^{2x}) &
			    \text{for $x \rightarrow - \infty$,} \\
			    {\cal O} (\re^{- 2x}) & \text{for $x \rightarrow + \infty$.}
                         \end{cases}
\end{equation}

For the calculation of the transformed jump matrix we remark that
\begin{equation}
     E_{p+} (x) - E_{p-} (x) = - \frac{2 \PH_{p-} (x) \re^{g(x)}}{1 - \re^{\eps(x)/T}}
\end{equation}
and
\begin{equation}
     f(x) = \frac{2 \re^{-g(x)}}{\PH_{p+} (x) (1 - \re^{- \eps(x)/T})} \epp
\end{equation}
Moreover, it is easy to see that
\begin{equation}
     \frac{\PH_{p-} (x)}{\PH_{p+} (x)} =
        \biggl(\frac{1 - \re^{- \eps(x)/T}}{1 + \re^{- \eps(x)/T}}\biggr)^2 \epp
\end{equation}
\end{proof}
\section{High-temperature analysis of the matrix Riemann-Hilbert
problem} \label{sec:hightmrhp}
\subsection{High-temperature form of the jump matrix}
\begin{proposition}
The functions occurring in the jump matrix $G_{\widetilde \chi}$
have the high-temperature expansions
\begin{subequations}
\begin{align} \label{highTA1}
     \frac{2}{1 - \re^{\pm \eps (x)/T}} & =
        \mp \frac{2T \re^{\mp \eps(x)/T}}{\eps(x)}
	    \bigl(1 + {\cal O} \bigl( T^{- 2} \bigr)\bigr) \epc \\[1ex]
	    \label{highTPhp}
     \PH_p (x) & = - \frac{\i p'(x)}{2} \tgh\biggl(\frac{h}{2T}\biggr)
                     \biggl(\cth^2\biggl(\frac{\eps(x)}{2T}\biggr)
		            \biggr)^{\one_{x \in \Int({\cal C}_p)}}
			    \notag \\ & \mspace{216.mu} \times
		     \frac{\sh^2(x - \la_F^+)}{\ch^2(x)}
		     \bigl(1 + {\cal O} \bigl( T^{- 2} \bigr)\bigr) \epp
\end{align}
\end{subequations}
\end{proposition}
\begin{proof}
For the proof of (\ref{highTPhp}) we rewrite the integral in the exponent
in (\ref{php}) as
\begin{multline}
     - \i \int_{{\cal C}_p} \rd \la \: p'(\la) z(\la)
                \frac{\sh(x + \la)}{\sh(x - \la)} = \\
     - 2 \int_{{\cal C}_p} \rd \la \: z(\la) \cth(2\la)
     + 2 \int_{{\cal C}_p} \rd \la \: z(\la) \cth(\la - x) \epp
\end{multline}
The first integral on the right hand side can be computed using the
symmetry of the integrand,
\begin{equation}
     - 2 \int_{{\cal C}_p} \rd \la \: z(\la) \cth(2\la)
        = \ln \biggl(\tgh\biggl(\frac{h}{2T}\biggr)\biggr) \epp
\end{equation}
The second integral we rewrite as
\begin{equation}
     2 \int_{{\cal C}_p} \rd \la \: z(\la) \cth(\la - x) =
        4 \p \i z(x) \one_{x \in \Int({\cal C}_p)} +
        2 \int_{{\cal C}_p'} \rd \la \: z(\la) \cth(\la - x) \epc
\end{equation}
where ${\cal C}_p'$ is a deformation of ${\cal C}_p$ in such a way
that $x \notin \Int\bigl({\cal C}_p'\bigr)$. Then we insert the
high-temperature expansion
\begin{equation}
     z(\la) = \frac{1}{2 \p \i} \bigl(\ln(2T) - \ln \bigl(\eps(\la)\bigr)\bigr)
              + {\cal O} \bigl( T^{- 2} \bigr)
\end{equation}
into the integral over ${\cal C}_p'$ and calculate the integral
using the fact that $\eps$ has a single simple zero inside
${\cal C}_p'$ at $z = \la_F^+$ and a single simple pole at
$z = \p \i/2$. Thus,
\begin{multline}
     2 \int_{{\cal C}_p'} \rd \la \: z(\la) \cth(\la - x) = 
     - \int_{{\cal C}_p'} \frac{\rd \la}{\p \i}
             \: \ln\bigl(\eps(\la)\bigr) \cth(\la - x)
             + {\cal O} \bigl( T^{- 2} \bigr) = \\
     2 \int_{\p \i/2}^{\la_F^+} \cth(\la - x) + {\cal O} \bigl( T^{- 2} \bigr) =
     \ln\biggl(- \frac{\sh^2(x - \la_F^+)}{\ch^2(x)}\biggr)
             + {\cal O} \bigl( T^{- 2} \bigr) \epp
\end{multline}
\end{proof}
In order to proceed we define a function
\begin{equation} \label{defalpha}
     \a (x) = \begin{cases}
              \frac{\sh(x - \la_F^+)}{\ch(x)} \re^{- (1/2T + \i t)(h + 2\i J \tgh(x))} &
	         \text{for $x \in \Ext {\cal C}_ p$} \\[1ex]
	      2 \re^{2x} \sh(x) \sh(x - \la_F^-) \re^{- (1/2T + \i t) 2\i J \cth(x)} &
	         \text{for $x \in \Int {\cal C}_ p$} \epc
              \end{cases}
\end{equation}
which is analytic and non-zero in ${\cal S} \setminus {\cal C}_p$ and
for large real part has the limits
\begin{equation} \label{alphapm}
     \a^{(\pm)} = \lim_{\Re x \rightarrow \pm \infty} \a (x)
              = \pm \re^{\mp \la_F^+ - (1/2T + \i t)(h \pm 2\i J)} \epp
\end{equation}
Using (\ref{highTPhp}) and (\ref{defalpha}) we obtain the high-temperature
expansion of the jump matrix $G_{\widetilde \chi}$,
\begin{equation} \label{gchitildeht}
     G_{\widetilde \chi} (x) =
        \begin{pmatrix}
	     1 & \frac{\a_- (x)}{\a_+ (x)} \bigl( \i \cth(x)\bigr)^{- m - 1} \\
	     - \frac{\a_+ (x)}{\a_- (x)} \bigl( \i \cth(x)\bigr)^{m + 1} & 0
        \end{pmatrix}
	   + {\cal O} \bigl( T^{- 2} \bigr) \epc
\end{equation}
where $\a_+ (x)$ and $\a_-(x)$ are the boundary values of $\a$ from inside
and outside ${\cal C}_p$.
\subsection{Transformation to lower triangular form}
For all $x \in {\cal S} \setminus {\cal C}_p$ we set
\begin{equation} \label{defpsi}
     \Ps (x) = \bigl(\a^{(+)}\bigr)^{- \s^z} \widetilde \chi (x)
               \bigl(\i \s^y\bigr)^{\one_{x \in \Int ({\cal C}_p)}}
	       \bigl(\a(x)\bigr)^{\s^z} \epc
\end{equation}
where $\one_{\rm condition}$ is the indicator function which is equal to
one if condition is satisfied and equal to zero else.  Then $\Ps$ solves
the matrix Riemann-Hilbert problem $({\cal S}, {\cal C}_p, + \infty, G_\Ps)$,
where, for all $x \in {\cal C}_p$, $G_\Ps$ is defined as
\begin{equation}
     G_\Ps (x) = G_\PH (x) + \de G_\Ps (x)
\end{equation}
with
\begin{equation} \label{defgph}
     G_\PH (x) = \begin{pmatrix}
                    \bigl(- \i \tgh(x)\bigr)^{- m - 1} & 0 \\
		    \a_+ (x) \a_- (x) & \bigl(- \i \tgh(x)\bigr)^{m + 1}
                 \end{pmatrix} \epc \qd
     \de G_\Ps (x) = {\cal O} \bigl( T^{-2} \bigr) \epp
\end{equation}
\subsection{Solution of the triangular model Riemann-Hilbert problem}
Consider the model matrix Riemann-Hilbert problem $({\cal S}, {\cal C}_p,
+ \infty, G_\PH)$ with jump matrix $G_\PH$ defined in (\ref{defgph}).
If it has a solution, then the solution is unique according to
theorem~\ref{th:mrhpunique}, since $G_\PH$ is unimodular. It can be
related to classical work of Fokas, Its and Kitaev \cite{FIK92} by
the change of variables
\begin{equation} \label{defzf}
     \fz (x) = - \i \tgh(x) = \i \, \frac{1 - \re^{2x}}{1 + \re^{2x}} \epp
\end{equation}
The map $x \mapsto \fz(x)$, ${\cal S} \mapsto {\mathbb C} P^1$ is
biholomorphic which is clear from the fact that we may consider it
as the composition of the exponential map $x \mapsto u = \re^{2x}$
with the Moebius transformation $u \mapsto z = \frac{\i - \i u}{1 + u}$.
It maps
\begin{equation} \label{pointsmap}
     (+ \infty, - \infty, \i \p/4, - \i \p/4, 0, \i \p/2) \mapsto
     (- \i, \i, 1, -1, 0, \infty) \epp
\end{equation}
Let $\G_p = \fz ({\cal C}_p)$. Then (\ref{pointsmap}) implies that
$\Int (\G_p) = \fz (\Ext ({\cal C}_p))$, $\Ext (\G_p) = \fz (\Int ({\cal C}_p))$.
Thus, $\G_p$ is a clockwise oriented closed contour enclosing the
points $0, 1, \i, - 1, - \i$. The point $\infty$ is outside $\G_p$.

Now let $z = \fz (x)$ and
\begin{equation} \label{defm}
     X(z) = \PH (x) \epc \qd M_m (z) = z^{- m - 1} (\a_+ \a_-) \bigl(\fz^{-1} (z)\bigr) \epp
\end{equation}
Then $X(z)$ satisfies the matrix Riemann-Hilbert problem
\begin{enumerate}
\item
$X(z)$ is holomorphic in ${\mathbb C}P^1$,
\item
\begin{equation}
     X(z) = \begin{cases}
	       I_2 + {\cal O} (z + \i) & \text{for $z \rightarrow - \i$} \\
	       X_\infty \bigl(I_2 + {\cal O} (z^{-1})\bigr) & \text{for $z \rightarrow \infty$,}
            \end{cases}
\end{equation}
\item
For $z \in \G_p$
\begin{equation}
     X_- (z) = X_+ (z) \begin{pmatrix}
                          z^{- m - 1} & 0 \\ z^{m+1} M_m (z) & z^{m+1}
                       \end{pmatrix} \epp
\end{equation}
\end{enumerate}
This matrix Riemann-Hilbert problem maps to the matrix Riemann-Hilbert
problem $({\cal S}, {\cal C}_p, + \infty, G_\PH)$ under $\fz$. Notice
that $X_+ (z)$ is the boundary value from \emph{outside} the
\emph{clockwise} oriented contour $\G_p$.

Finally the transformation
\begin{equation} \label{defy}
     Y(z) = X_\infty^{-1} X(z) \bigl\{z^{- (m + 1)\s^z} \one_{z \in \Ext(\G_p)}
               + I_2 \one_{z \in \Int(\G_p)} \bigr\}
\end{equation}
maps the matrix Riemann-Hilbert problem for $X$ onto an exactly solvable
matrix Riemann-Hilbert problem:
\begin{enumerate}
\item
$Y$ is holomorphic in ${\mathbb C} \setminus \G_p$,
\item
\begin{equation} \label{asyy}
     Y(z) = \bigl(I_2 + {\cal O} (z^{-1})\bigr) z^{-(m+1) \s^z} \epc
            \qd \text{for $z \rightarrow \infty$,}
\end{equation}
\item
$Y$ admits $\pm$ boundary values on $\G_p$ defining two continuous
functions $Y_\pm$ on $\G_p$ such that
\begin{equation} \label{jumpy}
     Y_- (z) = Y_+ (z) \begin{pmatrix} 1 & 0 \\ M_m (z) & 1 \end{pmatrix} \epp
\end{equation}
\end{enumerate}
This is a matrix Riemann-Hilbert problem considered in a classical
paper by Fokas, Its and Kitaev \cite{FIK92}.

Following this work we obtain its solution in terms of certain
explicitly constructible polynomials. Define the Cauchy transform
of a piecewise continuous function $f: {\cal C}_p \mapsto
{\mathbb C}$ as
\begin{equation}
    C_{\G_p} [f] (z) = \int_{\G_p} \frac{\rd w}{2 \p \i} \frac{f(w)}{w - z} \epp 
\end{equation}
Then the unique solution of the above matrix Riemann-Hilbert
problem is (see Appendix~\ref{app:modelmrhp})
\begin{equation} \label{solvey}
     Y(z) = \begin{pmatrix}
               - \g_m C_{\G_p} [M_m P_m] (z) & \g_m P_m (z) \\[1ex]
	       - C_{\G_p} [M_m Q_{m+1}] (z) & Q_{m+1} (z)
            \end{pmatrix} \epc
\end{equation}
where
\begin{equation} \label{defgammam}
     \g_m^{-1} = \int_{\G_p} \frac{\rd w}{2 \p \i} M_m (w) w^m P_m (w)
\end{equation}
and $P_m$ and $Q_{m+1}$ are monic polynomials of degrees $m$ and $m+1$.
These are defined by the conditions
\begin{subequations}
\label{ortho}
\begin{align} \label{orthop}
     \int_{\G_p} \frac{\rd w}{2 \p \i} M_m (w) w^n P_m (w) & = 0 \epc
       \qd \text{for $n = 0, 1, \dots, m - 1$} \epc \\[1ex]
     \int_{\G_p} \frac{\rd w}{2 \p \i} M_m (w) w^n Q_{m+1} (w) & = 0 \epc
       \qd \text{for $n = 0, 1, \dots, m$} \epp
\end{align}
\end{subequations}

More explicitly, let
\begin{equation} \label{defalk}
     A^l_k (m) =  \int_{\G_p} \frac{\rd w}{2 \p \i} M_m (w) w^{l+k} \epc \qd
     \text{for $k = 0, 1, \dots, m + 1$; $l = 0, \dots, m$.}
\end{equation}
Then equations (\ref{ortho}) are equivalent to saying that the coefficients
in the expansions
\begin{equation} \label{pqcoeff}
     P_m (z) = z^m + \sum_{k=0}^{m-1} c_m^{(k)} z^k \epc \qd
     Q_{m+1} (z) = z^{m+1} + \sum_{k=0}^{m} d_{m+1}^{(k)} z^k
\end{equation}
satisfy the linear equations
\begin{equation} \label{pqlineqs}
     \sum_{k=0}^{m-1} A_k^l (m) c_m^{(k)} = - A_m^l (m) \epc \qqd
     \sum_{k=0}^m A_k^n (m) d_{m+1}^{(k)} = - A_{m+1}^n (m)
\end{equation}
$l = 0, \dots, m - 1$; $n = 0, \dots, m$. Introducing the column vectors
\begin{equation}
     \Cv_i = (A^0_k (m), \dots, A^{m-1}_k (m))^t \epc \qd
     \Dv_j = (A^0_l (m), \dots, A^m_l (m))^t
\end{equation}
for $i = 0, \dots, m$, $j = 0, \dots, m + 1$, we can express the coefficients
of the polynomials $P_m$ and $Q_{m+1}$ by means of Cramer's rule,
\begin{subequations}
\label{cdcramer}
\begin{align}
     & c_m^{(n)} = - \frac{\det (\Cv_0, \dots, \Cv_{n-1}, \Cv_m, \Cv_{n+1},
                             \dots, \Cv_{m-1})}
			    {\det_{i, j = 0, \dots, m - 1} \bigl(A^i_j (m)\bigr)} \epc \\[1ex]
     & d_{m+1}^{(n)} = - \frac{\det (\Dv_0, \dots, \Dv_{n-1}, \Dv_{m+1}, \Dv_{n+1},
                             \dots, \Dv_{m})}
			    {\det_{i, j = 0, \dots, m} \bigl(A^i_j (m)\bigr)} \epc
			    \label{dcramer}
\end{align}
\end{subequations}
whenever the determinants in the denominator are non-zero. In
Appendix~\ref{app:expq} we show that the latter is the case for all
complex values of $\tau = - 4J \bigl(t - \frac{\i}{2T}\bigr)$ close
enough to the real axis, with the possible exception of finitely many
points.

In the same appendix we also obtain the following explicit expression
for the $A^l_k (m)$ as finite sums of modified Bessel functions $I_n (x)$:
\begin{multline} \label{alkmodbesselupper}
     A^l_k (m) = - \frac{\i^{m - k - l}\re^{\i c \tau}}{c} \biggl\{
                 I_{m - k - l - 1} (-\tau) \\ + \bigl(1 + \i c(m - k - l - 1)\bigr)
		 \biggl[\re^{- \tau} - \sum_{n = k + l + 1 - m}^{m - k - l - 1} I_n (- \tau)
		        \biggr] \\
		 + \i c \tau \bigl(I_{m - k - l - 1} (- \tau) + I_{m - k - l} (- \tau)\bigr)
		   \biggr\}
\end{multline}
for $m - k - l - 1 \ge 0$, and
\begin{multline} \label{alkmodbessellower}
     A^l_k (m) = - \frac{\i^{m - k - l}\re^{\i c \tau}}{c} \biggl\{
                 I_{m - k - l - 1} (-\tau) \\ + \bigl(1 + \i c(m - k - l - 1)\bigr)
		 \biggl[\re^{- \tau} + \sum_{n = m - k - l}^{k + l - m} I_n (- \tau) \biggr] \\
		 + \i c \tau \bigl(I_{m - k - l - 1} (- \tau)
		                   + I_{m - k - l} (- \tau)\bigr) \biggr\}
\end{multline}
for $m - k - l - 1 < 0$. Here we have employed the shorthand notation
\begin{equation}
     c = \cos (p_F) = \frac{h}{4J} \epp
\end{equation}

\subsection{Back transformation}
Going backwards we obtain, for all values of $t$ for which $Y$ exists,
\begin{equation}
     \PH(x) = Y^{-1} (- \i) Y(\fz (x))
              \bigl\{ \fz(x)^{(m+1) \s^z} \one_{x \in \Int {\cal C}_p}
	              + I_2 \one_{x \in \Ext {\cal C}_p} \bigr\} \epp
\end{equation}
This is the unique solution of the model matrix Riemann-Hilbert
problem $({\cal S}, {\cal C}_p, + \infty, G_\PH)$. In order to
relate it to $\Ps$, equation (\ref{defpsi}), we set
\begin{equation}
     \Ps (x) = \Upsilon (x) \PH (x)
\end{equation}
for all $x \in {\cal S} \setminus {\cal C}_p$. Then $\Upsilon$ solves
the matrix Riemann-Hilbert problem $({\cal S}, {\cal C}_p, + \infty,
G_\Upsilon)$ with jump matrix
\begin{equation} \label{defgups}
     G_\Upsilon (x) = I_2 + \PH_+ (x) \de G_\Ps (x) \PH_- (x) \epp
\end{equation}

It is a well-know fact \cite{BeCo84} that $\Upsilon$ is equivalently
determined as the solution of the singular integral equation
\begin{equation} \label{lieups}
     \Upsilon (x) = I_2 - \int_{{\cal C}_p} \frac{\rd \la}{2\p\i}
                             \ph(x,\la) \Upsilon_+ (\la)
			     \bigl(G_\Upsilon (\la) - I_2\bigr)
\end{equation}
which first determines $\Upsilon_+$ on ${\cal C}_p$ and then
$\Upsilon$ on ${\cal S} \setminus {\cal C}_p$ with correct
asymptotics by construction.  Now, since (\ref{defgph}), (\ref{defgups})
imply that $G_\Upsilon (\la) - I_2 = {\cal O} (T^{-2})$ and thus
\begin{equation}
     \| G_\Upsilon (\la) - I_2 \|_{{\rm Mat}_{2 \times 2} (L^q ({\cal C}_p))}
        = {\cal O} (T^{-2})
\end{equation}
for all $q > 1$, it follows from the results of \cite{Calderon77} that
the integral equation (\ref{lieups}) is solvable in terms of its
Neumann series. Since, furthermore, the jump matrix $G_\Upsilon$
admits analytic continuations to either ``$+$'' or ``$-$'' neighbourhoods
of ${\cal C}_p$ and similar estimates as above hold on any closed
curve in these neighbourhoods, it follows that
\begin{equation} \label{upshight}
     \Upsilon (x) = I_2 + {\cal O} (T^{-2})
\end{equation}
uniformly in ${\mathbb C}$ with a differentiable remainder. Thus, we have
arrived at the following
\begin{proposition}
The matrix Riemann-Hilbert problem $({\cal S}, {\cal C}_p, + \infty,
G_{\widetilde \chi})$ with jump matrix (\ref{transg}) has the
unique solution
\begin{multline} \label{solvechitilde}
     \widetilde \chi (x) = \bigl(\a^{(+)}\bigr)^{\s^z} \Upsilon (x)
        Y^{-1} (- \i) Y(\fz (x)) \\ \times
	\bigl\{ \fz(x)^{(m+1) \s^z} \one_{x \in \Int {\cal C}_p}
                + I_2 \one_{x \in \Ext {\cal C}_p} \bigr\}
		\a(x)^{- \s^z} \bigl(- \i \s^y\bigr)^{\one_{x \in \Int {\cal C}_p}}
\end{multline}
with $\a$ from (\ref{defalpha}), $\fz$ from (\ref{defzf}), $Y$ from (\ref{solvey}) and
$\Upsilon$ from (\ref{lieups}), which exists for all complex times $t$ close
enough to the real axis with the possible exception of finitely many values.
\end{proposition}

\section{High-T expansion of the transverse auto-correlation
function of the XX chain} \label{sec:corrhighT}
In this section we use propositions~\ref{prop:resolvepart} and
\ref{prop:dtlnfred} in order to calculate the leading order
high-temperature asymptotics of the transverse dynamical correlation
function of the XX chain. We shall reproduce and generalize the
Brandt-Jacoby formula (\ref{brandtjacoby}).
\begin{proposition}
The `resolvent part' of our Fredholm determinant representation
(\ref{freddetrep}) of the transverse auto-correlation function
has the large-$T$ behaviour
\begin{multline} \label{resolveasy}
     \Om - \int_{{\cal C}_p} \rd x \:
              w(x) \re^{- g(x)} E_h (x) (\id - \widehat R) E_h (x) = \\
	\i \bigl(\a^{(+)}\bigr)^2 \g_m \bigl(Q_{m+1} (- \i) P_m' (- \i)
	   - P_m (- \i) Q_{m+1}' (- \i)\bigr)
	+ {\cal O} \bigl(T^{-2}\bigr) \epp
\end{multline}
\end{proposition}
\begin{proof}
Using that $\lim_{\Re x \rightarrow + \infty} \widetilde \chi (x) = I_2$
we first obtain
\begin{multline}
      \lim_{\Re x \rightarrow + \infty} \re^{2x} \chi^{12} (x) =
         - \lim_{\Re x \rightarrow + \infty} E_p (x) \\
	 + \lim_{\Re x \rightarrow + \infty}
	   \bigl\{ \re^{2x} \bigl(\a^{(+)}\bigr)^{\s^z} \Upsilon (x)
	           Y^{-1} (- \i) Y(\fz(x)) \a(x)^{- \s^z} \bigr\}^{12}
\end{multline}
In the second term on the right hand side we use that
\begin{equation}
     \Upsilon (x) = I_2 + \Upsilon^{(+)}_1 \re^{- 2x}
                    + {\cal O} \bigl(\re^{- 4x}\bigr) \epc \qd
     \fz (x) = - \i + 2 \i \, \re^{- 2x} + {\cal O} \bigl(\re^{- 4x}\bigr)
\end{equation}
and
\begin{equation}
      \bigl[ \bigl(\a^{(+)}\bigr)^{\s^z} \Upsilon^{(+)}_1
         \bigl(\a^{(+)}\bigr)^{- \s^z} \bigr]^{12} = {\cal O} \bigl(T^{- 2}\bigr) \epp
\end{equation}
We further observe that
\begin{equation}
     \lim_{\Re x \rightarrow + \infty} E_p (x) = 2 \Om \epp
\end{equation}
Then
\begin{equation}
      \lim_{\Re x \rightarrow + \infty} \re^{2x} \chi^{12} (x) =
         - 2 \Om + 2\i \bigl(\a^{(+)}\bigr)^2 \bigl(Y^{-1} (- \i) Y'(- \i)\bigr)^{12}
	 + {\cal O} \bigl(T^{- 2}\bigr)
\end{equation}
and the claim follows with equation (\ref{solvey}) and proposition~\ref{prop:resolvepart}.
\end{proof}
In order to calculate the Fredholm determinant factor in (\ref{freddetrep})
we define
\begin{equation} \label{defds}
     D(x) = e^1_1 \one_{x \in \Int({\cal C}_p)} +
            e^2_2 \one_{x \in \Ext({\cal C}_p)} \epc \qd
     s(x) = \one_{x \in \Ext({\cal C}_p)} - \one_{x \in \Int({\cal C}_p)} \epp
\end{equation}
We start by deriving the high-temperature asymptotics of the integrand
in (\ref{dtlnfredmrhp}).
\begin{proposition} \label{prop:trchischiht}
For high temperatures the integrand in (\ref{dtlnfredmrhp}) behaves as
\begin{multline} \label{dtlnintegrand}
     \tr\bigl\{ \chi' (x) \bigl(e_2^2 + \re^{- 2x} E_p (x) e_1^2\bigr)
                \chi^{-1} (x) \bigr\} = \frac{\a' (x)}{\a(x)} s(x) \\[1ex]
		+ \frac{2 (m+1)}{\sh(2x)} \one_{x \in \Int {\cal C}_p}
		+ \fz' (x) \tr \bigl\{D(x) Y^{-1} (\fz(x)) Y'(\fz(x))\bigr\}
		     + {\cal O} \bigl(T^{-2}\bigr) \epp
\end{multline}
\end{proposition}
\begin{proof}
It can be easily checked that
\begin{equation}
     \tr\bigl\{ \chi' (x) \bigl(e_2^2 + \re^{- 2x} E_p (x) e_1^2\bigr)
                \chi^{-1} (x) \bigr\} = 
     \tr\bigl\{ \widetilde \chi' (x) e_2^2 \widetilde \chi^{-1} (x) \bigr\} \epp
\end{equation}
Inserting (\ref{upshight}) and (\ref{solvechitilde}) on the right hand
side and using the definitions (\ref{defds}) we have proved the claim.
\end{proof}

Proposition~\ref{prop:trchischiht} can be used to obtain the high-temperature
asymptotics of the Fredholm determinant contribution to the transverse
auto-correlation function by means of proposition~\ref{prop:dtlnfred}.
Employing the short-hand notation
\begin{equation} \label{deffg}
     F_m (x) = C_{\G_p} [M_m P_m] (x) \epc \qd
     G_{m+1} (x) = C_{\G_p} [M_m Q_{m+1}] (x) \epc
\end{equation}
in (\ref{solvey}) we can formulate
\begin{proposition} \label{prop:dtlndetasy}
The logarithmic derivative of the determinant part of the Fredholm determinant
representation (\ref{freddetrep}) of the transverse auto-correlation function
has the large-$T$ asymptotic behaviour
\begin{multline} \label{dtlnfredasy}
     \6_t \ln\Bigl(\det_{{\cal C}_p} \bigl(\id + \widehat V\bigr)\Bigr) =
        - 8 J^2 t - 4J + \frac{\i 4 J^2}{T} + \6_t \ln \bigl( \g_m^{-1} \re^{\i h t} \bigr)
	- 2 \i J c_m^{(m-1)} \\
	+ 2 \i J \g_m \bigl[F_m (0)\bigl(P_m (0) - Q_{m+1}' (0)\bigr)
	                    + G_{m+1} (0) P_m' (0)\bigr]
	+ {\cal O} \bigl(T^{-2}\bigr) \epp
\end{multline}
\end{proposition}
\begin{proof}
We insert (\ref{dtlnintegrand}) into (\ref{dtlnfredmrhp}).
Three integrals corresponding to the three terms on the right hand side of
(\ref{dtlnintegrand}) remain to be calculated. The integration contour
$\g ({\cal C}_p)$ closely surrounds the particle contour ${\cal C}_p$.
We decompose it into a sum of an exterior part ${\cal C}_{\rm ext}$ and
an interior part ${\cal C}_{\rm int}$ in such a way that $\g ({\cal C}_p)
= {\cal C}_{\rm ext} - {\cal C}_{\rm int}$.

The first integral is
\begin{multline} \label{firstintdtlnfred}
     I_1 = \int_{\g ({\cal C}_p)} \frac{\rd x}{2\p} \: \eps(x) s(x) \6_x \ln (\a(x)) =
     \int_{{\cal C}_{\rm ext}} \frac{\rd x}{2\p} \: \eps(x)
          \biggl[\cth(x - \la_F^+) - \tgh(x) - \frac{\tau}{2 \ch^2 (x)}\biggr] \\[1ex]
     + \int_{{\cal C}_{\rm int}} \frac{\rd x}{2\p} \: \eps(x)
          \biggl[2 + \cth(x - \la_F^-) + \cth(x) + \frac{\tau}{2 \sh^2 (x)}\biggr] \\[1ex]
     = - 8 J^2 t - 4J + \frac{\i 4 J^2}{T}.
\end{multline}
The integrals can be easily calculated by means of the residue theorem, since
the term in square brackets under the integral over ${\cal C}_{\rm ext}$ is
holomorphic outside ${\cal C}_p$, $\p \i$-periodic  and ${\cal O} (\re^{-2|\Re z|})$,
while the term in square brackets under the integral over ${\cal C}_{\rm int}$
is holomorphic inside ${\cal C}_p$.

The second integral is
\begin{equation}
     I_2 = - \int_{{\cal C}_{\rm int}} \frac{\rd x}{2\p} \: \eps(x) \frac{2 (m+1)}{\sh(2x)}
        = \i h (m+1) \epp
\end{equation}
For the third integral we perform the change of variables $x \mapsto z = \fz(x)$.
Then ${\cal C}_{\rm out} \mapsto \G_{\rm int}$, ${\cal C}_{\rm int} \mapsto \G_{\rm out}$,
where $\G_{\rm int}$ and $\G_{\rm out}$ are simple clockwise oriented contours, and
\begin{multline}
     I_3 = \int_{\g ({\cal C}_p)} \frac{\rd x}{2\p} \: \eps(x)
		 \fz' (x) \tr \bigl\{D(x) Y^{-1} (\fz(x)) Y'(\fz(x))\bigr\} \\
         = \int_{\G_{\rm int}} \frac{\rd z}{2\p} \: \Bigl[h - 2J \Bigl(z + \frac{1}{z}\Bigr)\Bigr]
                 \tr \bigl\{e_2^2 Y^{-1} (z) Y'(z)\bigr\} \\
         - \int_{\G_{\rm out}} \frac{\rd z}{2\p} \: \Bigl[h - 2J \Bigl(z + \frac{1}{z}\Bigr)\Bigr]
                 \tr \bigl\{e_1^1 Y^{-1} (z) Y'(z)\bigr\} \epp
\end{multline}
Here the integrals on the right hand side can be calculated by taking the
residues at $z = 0$ and at $z = \infty$. Taking into account that
$Y^{-1} (z) Y'(z)$ is holomorphic inside $\G_{\rm int}$ and that, being
a logarithmic derivative,
\begin{equation} \label{eppyypasy}
     \tr \bigl\{e_1^1 Y^{-1} (1/z) Y'(1/z)\bigr\} = \be_1 z + \be_2 z^2 + {\cal O} (z^3)
\end{equation}
for $z \rightarrow 0$, we obtain
\begin{equation}
     I_3 = 2 \i J \tr \bigl\{e_2^2 Y^{-1} (0) Y'(0)\bigr\} + \i h \be_1 - 2 \i J \be_2 \epp
\end{equation}

Hence, the trace on the right hand side and the coefficients $\be_1$, $\be_2$
remain to be calculated. Writing
\begin{equation}
     Y(z) = \begin{pmatrix} a & b \\ c & d \end{pmatrix}
\end{equation}
for short and using the unimodularity of $Y$ we see that
\begin{equation}
     Y^{-1} (z) Y'(z) = \begin{pmatrix}
			   a' d - b c' & \ast \\ \ast & a d' - b' c
                        \end{pmatrix} \epp
\end{equation}
It follows that
\begin{equation} \label{traceless}
     \tr\{Y^{-1} (z) Y'(z)\} = \6_z (ad - bc) = \6_z \det Y(z) = 0
\end{equation}
and
\begin{multline}
     \tr \bigl\{e_2^2 Y^{-1} (0) Y'(0)\bigr\} = (a d' - b' c)|_{z = 0} \\
        = - \g_m \bigl(F_m(0) Q_{m+1}'(0) - G_{m+1}(0) P_m'(0)\bigr) \epp
\end{multline}

The asymptotic behaviour of $P_m$, $Q_{m+1}$, $F_m$ and $G_{m+1}$ for
large $|z|$ can be read of from (\ref{pqcoeff}), (\ref{deffg}), implying
first of all that
\begin{multline} \label{eppyyp}
     \tr \bigl\{e_1^1 Y^{-1} (1/z) Y'(1/z)\bigr\} = (a' d - b c')|_{z \rightarrow 1/z} \\
        = - \g_m F_m' (1/z) Q_{m+1} (1/z) + {\cal O} (z^3)
\end{multline}
and further
\begin{subequations}
\label{qfzinv}
\begin{align}
     Q_{m+1} (1/z) & = z^{- m - 1} + d_{m+1}^{(m)} z^{-m} + {\cal O} (z^{- m + 1}) \epc \\[1ex]
     F_m (1/z) & = - \g_m^{-1} z^{m + 1} - \k_m z^{m + 2} + {\cal O} (z^{m + 3}) \epc
\end{align}
\end{subequations}
where
\begin{equation}
     \k_m = \int_{\G_p} \frac{\rd w}{2 \p \i} M_m (w) w^{m+1} P_m (w) \epp
\end{equation}
Inserting (\ref{qfzinv}) into (\ref{eppyyp}) and comparing with (\ref{eppyypasy})
we conclude that
\begin{equation}
     \be_1 = - m - 1 \epc \qd
     \be_2 = - \bigl[(m+1) d_{m+1}^{(m)} + (m+2) \g_m \k_m\bigr] \epp
\end{equation}
Thus,
\begin{multline} \label{ithreefirst}
     I_3 = - 2 \i \g_m \bigl(F_m(0) Q_{m+1}'(0) - G_{m+1}(0) P_m'(0)\bigr)
           - \i h (m + 1) \\
	   + 2 \i J \bigl[(m+1) d_{m+1}^{(m)} + (m+2) \g_m \k_m\bigr] \epp
\end{multline}

Alternatively, due to (\ref{traceless}), we may repeat the calculation
replacing $(a' d - b c')|_{z \rightarrow 1/z}$ by
$- (a d' - b' c)|_{z \rightarrow 1/z}$ in (\ref{eppyyp}). We obtain
\begin{multline} \label{ithreesecond}
     I_3 = - 2 \i \g_m \bigl(F_m(0) Q_{m+1}'(0) - G_{m+1}(0) P_m'(0)\bigr)
           - \i h (m + 1) \\
	   + 2 \i J \bigl[m d_{m+1}^{(m)} + (m+1) \g_m \k_m\bigr] \epp
\end{multline}
Comparing (\ref{ithreefirst}) and (\ref{ithreesecond}) we conclude that
\begin{equation}
     d_{m+1}^{(m)} + \g_m \k_m = 0 \epp
\end{equation}
It follows that
\begin{multline} \label{ithreethird}
     I_3 = - 2 \i \g_m \bigl(F_m(0) Q_{m+1}'(0) - G_{m+1}(0) P_m'(0)\bigr)
           - \i h (m + 1) + 2 \i J \g_m \k_m \epp
\end{multline}
This can be slightly rewritten noticing that
\begin{align} \label{loggammaid}
     \re^{- \i h t} \6_t \re^{\i h t} \g_m^{-1} & =
        \frac{\i h}{\g_m} + \6_t \int_{\G_p} \frac{\rd w}{2 \p \i} M_m (w) w^m P_m (w) \notag \\[1ex]
        & = \frac{\i h}{\g_m} +
	    \int_{\G_p} \frac{\rd w}{2 \p \i} \bigl((\6_t M_m (w)) P_m^2 (w)
	                + 2 M_m (w) P_m (w) \6_t P_m (w) \bigr) \notag \\[1ex]
        & = 2 \i J \int_{\G_p} \frac{\rd w}{2 \p \i} M_m (w)
	                       \biggl(w - \frac{1}{w}\biggr) P_m^2 (w) \notag \\[1ex]
        & = 2 \i J \k_m + 2 \i J \frac{c_m^{(m-1)}}{\g_m} - 2 \i J P_m (0) F_m (0) \epp
\end{align}
Here we have used equation (\ref{orthop}) in the second, third and fourth
equation. Inserting (\ref{loggammaid}) into (\ref{ithreethird}) and adding
up $I_1$, $I_2$ and $I_3$ we arrive at the right hand side of (\ref{dtlnfredasy}).
\end{proof}

Propositions \ref{prop:trchischiht} and \ref{prop:dtlndetasy} combined
with the original form-factor series (\ref{ffseriesxxtrans}) are enough
to obtain the leading high-temperature asymptotics of the dynamical
correlation functions~(\ref{defcorrmp}). The result is most naturally
expressed in terms of the variable $\tau = - 4J \bigl(t - \frac{\i}{2T}\bigr)$
introduced above.

\begin{theorem} \label{theo:main}
In the high-$T$ limit the transverse dynamical correlation function of the
XX-chain behaves as
\begin{multline} \label{xxtransht}
     \bigl\<\s_1^- \s_{m+1}^+ (t)\bigr\>_T = \2 \biggl(- \frac{J}{T}\biggr)^m
        \exp\bigg\{\i c \tau  - \frac{\tau^2}{4}
	           + \int_0^\tau \rd \tau' \: u_m (\tau') \biggr\} \\[1ex] \times
	\frac{Q_{m+1}(- \i) P_m' (- \i) - P_m (- \i) Q_{m+1}' (- \i)}
	     {\bigl(Q_{m+1}(- \i) P_m' (- \i)
	            - P_m (- \i) Q_{m+1}' (- \i)\bigr)\bigr|_{\tau = 0}}
        \bigl(1 + {\cal O} (T^{-2}) \bigr) \epc
\end{multline}
where
\begin{equation} \label{defu}
     u_m (\tau) = \frac{\i}{2}
                  \bigl[c_m^{(m-1)} - \g_m \bigl\{F_m (0) (P_m (0) - Q_{m+1}'(0))
		        + G_{m+1} (0) P_m'(0) \bigr\} \bigr] \epp
\end{equation}
\end{theorem}
\begin{proof}
Expressing equation (\ref{dtlnfredasy}) in terms of $\tau$ and the
$\tau$-derivative and integrating over from $0$ to some fixed value
of $\tau$ we obtain
\begin{multline}
     \det_{{\cal C}_p} \bigl(\id + \widehat V\bigr) =
        \bigl[ \det_{{\cal C}_p} \bigl(\id + \widehat V\bigr) \g_m \bigr]_{\tau = 0} \\ \times
	\g_m^{-1} \exp\bigg\{- \i c \tau + \tau - \frac{\tau^2}{4}
	                     + \int_0^\tau \rd \tau' \: u_m (\tau') \biggr\}
        \bigl(1 + {\cal O} (T^{-2}) \bigr) \epc
\end{multline}
where the definition (\ref{defu}) was used as well. Inserting this
expression together with (\ref{resolveasy}) and (\ref{alphapm}) into
(\ref{freddetrep}) we obtain
\begin{multline} \label{xxtranshtafb}
     \bigl\<\s_1^- \s_{m+1}^+ (t)\bigr\>_T = \\ \i (-1)^m {\cal F} (m)
        \bigl[ \det_{{\cal C}_p} \bigl(\id + \widehat V\bigr) \g_m \bigr]_{\tau = 0}
	\exp\bigg\{- 2 \la_F^+ + \i c \tau - \frac{\tau^2}{4}
	           + \int_0^\tau \rd \tau' \: u_m (\tau') \biggr\} \\[1ex] \times
	\bigl(Q_{m+1}(- \i) P_m' (- \i) - P_m (- \i) Q_{m+1}' (- \i)\bigr)
        \bigl(1 + {\cal O} (T^{-2}) \bigr) \epp
\end{multline}
For $t = \frac{\i}{2T}$ this becomes
\begin{multline} \label{xxtranshtafbtnull}
     \bigl\<\s_1^- \s_{m+1}^+ \bigl(\tst{\frac{\i}{2T}}\bigr)\bigr\>_T
        = \i \re^{- 2 \la_F^+} (-1)^m {\cal F} (m)
        \bigl[ \det_{{\cal C}_p} \bigl(\id + \widehat V\bigr) \g_m \bigr]_{\tau = 0}
	   \\[1ex] \times
	\bigl(Q_{m+1}(- \i) P_m' (- \i) - P_m (- \i) Q_{m+1}' (- \i)\bigr)\bigr|_{\tau = 0}
        \bigl(1 + {\cal O} (T^{-2}) \bigr) \epp
\end{multline}
Setting $t = \frac{\i}{2T}$ in the form factor series (\ref{ffseriesxxtrans})
and expanding each term to leading order in $1/T$, on the other hand, we see that
\begin{equation}
     \bigl\<\s_1^- \s_{m+1}^+\bigl(\tst{\frac{\i}{2T}}\bigr)\bigr\>_T =
        \2 \biggl(- \frac{J}{T}\biggr)^m \bigl(1 + {\cal O} (T^{-2}) \bigr)
\end{equation}
which, together with (\ref{xxtranshtafb}) and (\ref{xxtranshtafbtnull}),
entails the claim.
\end{proof}
\begin{corollary}
In the high-$T$ limit the transverse auto-correlation function of the
XX-chain behaves as
\begin{equation} \label{xxtransautoht}
     \bigl\<\s_1^- \s_{1}^+ (t)\bigr\>_T = \2 \re^{- \i h(t - \i/(2T)) - 4 J^2 (t - \i/(2T))^2}
        \bigl(1 + {\cal O} (T^{-2}) \bigr) \epp
\end{equation}
\end{corollary}
\begin{proof}
For $m = 0$ we have $P_0 (z) = 1$, $Q_1 (z) = z + d_1^{(0)}$, implying that
\begin{equation}
     Q_1 (- \i) P_0' (- \i) - P_0 (- \i) Q_1' (- \i) = - 1
\end{equation}
and that $u_0 (\tau) = \i c_0^{(-1)}/2$. From (\ref{loggammaid}) we see
that the proper definition of the constant $c_0^{(-1)}$ is $c_0^{(-1)} = 0$.
Thus, $u_0 = 0$, and (\ref{xxtransautoht}) follows from (\ref{xxtransht}) for
$m = 0$.
\end{proof}
Equation (\ref{xxtransautoht}) reproduces the Brandt-Jacoby formula
(\ref{brandtjacoby}) for $T \rightarrow + \infty$. Looking at it the
other way round, we see that `switching on the inverse temperature'
means, to leading order, to replace $t$ by $t - \i/(2T)$.

\begin{remark}
A similar phenomenon can be observed in the longitudinal case (see e.g.\
\cite{GKKKS17}, equation (113)):
\begin{multline}
     \<\s_1^z \s_{m+1}^z (t)\>_T =
     \biggl[ \int_{- \p}^\p \frac{\rd p}{2 \p} \tgh \biggl(\frac{\e(p)}{2T}\biggr) \biggr]^2
     \\[1ex] +
        \biggl[ \int_{-\p}^\p \frac{\rd p}{2\p} \:
	   \frac{\re^{\i(mp - (t - \i/(2T))\e (p))}}{\ch \bigl(\e(p)/(2T)\bigr)} \biggr]
        \biggl[ \int_{-\p}^\p \frac{\rd p}{2\p} \:
	   \frac{\re^{- \i(mp - (t - \i/(2T))\e (p))}}{\ch \bigl(\e(p)/(2T)\bigr)} \biggr]
\end{multline}
with $\e(p) = h - 4J \cos(p)$, implying that
\begin{equation}
     \bigl\<\s_1^z \s_{m+1}^z (t)\bigr\>_T = J_m^2 \bigl(4J(t - \i/(2T))\bigr)
	\bigl(1 + {\cal O} (T^{-2}) \bigr) \epc
\end{equation}
where $J_m$ is a Bessel function.
\end{remark}
For the actual evaluation of the asymptotic formula (\ref{xxtransht})
one first has to compute $A^l_k (m)$ from (\ref{alkmodbesselupper})
and (\ref{alkmodbessellower}) for a given numerical value of $\tau$.
Then the coefficients $c_m^{(k)}$ and $d_{m+1}^{(l)}$ can be
efficiently computed by means of (\ref{pqlineqs}). This fixes $P_m$ and
$Q_{m+1}$ for a fixed value of $\tau$. For the computation of $u_m$
one can use the formulae
\begin{subequations}
\begin{align}
     & \g_m^{-1} = \sum_{k=0}^m A_k^m (m) c_m^{(k)} \epc \\
     & F_m (0) = \sum_{k=0}^m A_k^0 (m+1) c_m^{(k)} \epc \qd
       G_{m+1} (0) = \sum_{k=0}^{m+1} A_k^0 (m+1) d_{m+1}^{(k)} \epc \\[1ex]
     & P_m (0) = c_m^{(0)} \epc \qd
       P_m' (0) = c_m^{(1)} \epc \qd Q_{m+1}' (0) = d_{m+1}^{(1)} \epp
\end{align}
\end{subequations}

Alternatively, using a computer-algebra program, it is not difficult
to obtain explicit expressions for the above functions in terms of
modified Bessel functions. However, the size of these expressions grows
rapidly with $m$, especially when $h \ne 0$. For this reason we refrain
from providing an extensive list of examples. In Appendix~\ref{app:explicit_m}
we demonstrate that equation (\ref{xxtransht}) reproduces the leading
order of the high-T expansions for $h=0$ and $m=1, 2$ derived in \cite{PeCa77}
(the case $m=0$ having already been checked above). In addition, we derive
the following new explicit formula for $m = 3$,
\begin{multline} \label{examplem4}
     \bigl< \sigma_1^-  \sigma_4^+ (t) \bigr\>_T\bigr|_{ h=0}
        \sim 16 \Bigl( -\frac{J}{T} \Bigr)^3
          {\rm e}^{-4J^2 t^2} \, \frac{I_1(-4Jt)}{ (4Jt)^5} \\[1ex] \times
	    \bigl((4Jt)^2 I_0^2(-4Jt) + 4Jt  I_0(-4Jt) I_1(-4Jt)
		   - \bigl((4Jt)^2+2\bigr) I_1^2(-4Jt) \bigr) \epp
\end{multline}
The reader is encouraged to work out more examples.
 
\section{Conclusions} \label{sec:conclusions}
Together with \cite{GKSS19,GKS19bpp} this paper is part of a series of works
in which we reconsider the transverse dynamical correlation function of the
XX chain at finite temperature. Based on a novel thermal form-factor series
(\ref{ffseriesxxtrans}) we have derived a Fredholm determinant representation
that is manifestly different from the Fredholm determinant representation of
Colomo et al.\ \cite{CIKT92}. Our Fredholm determinant representation and
the associated matrix Riemann-Hilbert problem can be used to analyse the
transverse correlation function numerically \cite{GKSS19} and asymptotically,
either for long times and large distances \cite{GKS19bpp} or in the high-temperature
limit.  In this work we have concentrated on the high-temperature asymptotic
analysis and have generalized the classical result (\ref{brandtjacoby}) of
Brandt and Jacoby \cite{BrJa76} for the autocorrelation function at
infinite temperature to arbitrary separations of space-time points,
but also to include the first order corrections in $1/T$.

\vspace{.5ex}
\noindent {\bf Acknowledgements.}
The authors would like to thank Alexander Its and Nikita Slavnov
for helpful discussions. They are grateful to Jacques Perk
for his explanations on the literature. FG was supported by the Deutsche
Forschungsgemeinschaft within the framework of the research
unit FOR 2316 `Correlations in integrable quantum many-body systems'.
The work of KKK was supported by the CNRS and by the `Projet
international de coop\'eration scientifique No.\ PICS07877':
\textit{Fonctions de corr\'elations dynamiques dans la
cha\^{\nodoti}ne XXZ \`a temp\'erature finie}, Allemagne,
2018-2020. JS was supported by JSPS KAKENHI Grants, numbers 15K05208,
18K03452 and 18H01141.

\clearpage

{\appendix
\section{The model Riemann-Hilbert problem}
\label{app:modelmrhp} \vspace{-1ex} \noindent
The model matrix Riemann-Hilbert problem for the matrix function $Y$, see
equation (\ref{defy}) and below, was solved by Fokas, Its and Kitaev
\cite{FIK92}. For the sake of completeness we shall repeat their arguments
here.

First of all the jump condition (\ref{jumpy}) is equivalent to
\begin{equation}
     \begin{pmatrix}
        {Y_1^1}_+ (z) + M_m (z) {Y_2^1}_+ (z) - {Y_1^1}_- (z) &
        {Y_2^1}_+ (z) - {Y_2^1}_- (z) \\[1ex]
        {Y_1^2}_+ (z) + M_m (z) {Y_2^2}_+ (z) - {Y_1^2}_- (z) &
        {Y_2^2}_+ (z) - {Y_2^2}_- (z)
     \end{pmatrix}
        = 0 \epp
\end{equation}
This is satisfied if we choose $Y_2^a (z)$, $a = 1, 2$, as arbitrary
entire functions and
\begin{equation} \label{y1form}
     Y_1^a (z) = g^a (z) - \int_{\G_p} \frac{\rd w}{2 \p \i}
                                       \frac{M_m (w) Y_2^a (w)}{w - z} \epc
\end{equation}
where $g^a (z)$, $a = 1, 2$, are two more arbitrary entire functions.
The arbitrariness of $Y_2^a (z)$ and $g^a (z)$, $a = 1, 2$, is lifted
by imposing the asymptotic condition (\ref{asyy}) which reads more
explicitly
\begin{equation} \label{asyyexp}
     Y(z) = \begin{pmatrix}
	       z^{- m - 1} \bigl(1 + {\cal O} (z^{-1})\bigr) &
	       \g_m z^ m \bigl(1 + {\cal O} (z^{-1})\bigr) \\[1ex]
	       \widetilde \g_m z^{- m - 2} \bigl(1 + {\cal O} (z^{-1})\bigr) &
	       z^{m + 1} \bigl(1 + {\cal O} (z^{-1})\bigr)
            \end{pmatrix} \epc
\end{equation}
where $\g_m, \widetilde \g_m \in {\mathbb C}$ are constants. Thus,
by Liouville's theorem
\begin{equation}
\label{gy2form}
\begin{split}
     & g^a (z) = 0 \epc \qd a = 1, 2 \epc \\
     & Y_2^1 (z) = \g_m P_m (z) \epc \qd Y_2^2 (z) = Q_{m+1} (z) \epc
\end{split}
\end{equation}
where $P_m$ is a monic polynomial of degree $m$ and $Q_{m+1}$ is a monic
polynomial of degree $m + 1$. Inserting (\ref{gy2form}) into (\ref{y1form})
it further follows that
\begin{equation}
     Y_1^1 (z) = - \int_{\G_p} \frac{\rd w}{2 \p \i} \frac{M_m (w) \g_m P_m (w)}{w - z} \epc \
     Y_2^1 (z) = - \int_{\G_p} \frac{\rd w}{2 \p \i} \frac{M_m (w) Q_{m+1} (w)}{w - z} \epp
\end{equation}
Expanding these expressions asymptotically for large $|z|$ and comparing
once more with (\ref{asyyexp}) we obtain the `orthogonality conditions'
(\ref{ortho}), which fix $P_m$ and $Q_{m+1}$ uniquely for at least almost
all $t$ (cf.\ Appendix~\ref{app:expq}), and equation (\ref{defgammam}) for
$\g_m$.

\section{\boldmath Properties of the polynomials $P_m$ and $Q_{m+1}$}
\label{app:expq} \vspace{-1ex} \noindent
In this appendix we work out some of the properties of the polynomials
$P_m$ and $Q_{m+1}$. These polynomials are well-defined whenever the
determinants $\det_{i, j = 0, \dots, m - 1} \bigl(A^i_j (m)\bigr)$ and
$\det_{i, j = 0, \dots, m} \bigl(A^i_j (m)\bigr)$ do not vanish. Their
coefficients are then determined by (\ref{cdcramer}). In~\ref{app:alkexplicit}
we derive the explicit expressions (\ref{alkmodbesselupper}),
(\ref{alkmodbessellower}) for the $A^l_k (m)$. In~\ref{app:welldef} we
show that the two determinants do not vanish at $\tau = 0$. Being entire
functions of $\tau$ they can then at most vanish on a discrete
subset of ${\mathbb C}$ with an accumulation point at infinity.
In~\ref{app:alklarget} we study the $A^l_k (m)$ for $\tau \rightarrow
+ \infty$. This allows us to conclude that the determinants are non-zero
for $\tau$ in a neighbourhood of the real axis, where thus the
polynomials $P_m$ and $Q_{m+1}$ exist, for all but possibly finitely
many values of $\tau$. We also obtain explicit expressions for the
leading $\tau \rightarrow + \infty$ asymptotics of the polynomials that
allow us to estimate the behaviour of $Q_{m+1}(- \i) P_m' (- \i) -
P_m (- \i) Q_{m+1}' (- \i)$ for $t \rightarrow - \infty$.
\subsection{Explicit formulae for matrix elements}
\label{app:alkexplicit}
Setting
\begin{equation}
     f(w) = - \frac{\i}{2} \Bigl(w - \frac{1}{w}\Bigr) \epc  \qd
     b_m (w) = \frac{\re^{\i c \tau} w^{-m}}{2 \p c}
                 \biggl( \frac{w - \i}{w + \i} - \frac{2c w}{(w + \i)^2} \biggr)
\end{equation}
we can write (cf.\ (\ref{defm}), (\ref{defalk}))
\begin{equation} \label{alkrationalint}
     A^l_k (m) = \int_{\G_p} \rd w \: b_{m - k - l} (w) \re^{\tau f(w)} \epp
\end{equation}
We recall that $\G_p$ is a simple, clockwise oriented contour encircling
the points $0$ and $- \i$. The exponential factor under the integral
is equal to the generating function of the Bessel functions of the first
kind,
\begin{equation} \label{besselgenfun}
     \re^{\tau f(w)} = \sum_{n \in {\mathbb Z}} J_n (- \i \tau) w^n \epp
\end{equation}
Inserting this into (\ref{alkrationalint}) and using the residue theorem
we conclude that
\begin{multline}
     A^l_k (m) = - \frac{\i^{m - k - l}\re^{\i c \tau}}{c} \biggl\{
                 J_{m - k - l - 1} (- \i \tau) \, \i^{k + l + 1 - m} \\
		 + 2 \bigl(1 + \i c(m - k - l - 1)\bigr) \mspace{-12.mu}
		     \sum_{n = m - k - l}^\infty \mspace{-12.mu} J_n (- \i \tau) \, \i^{-n}
                 - 2 \i c  \mspace{-12.mu} \sum_{n = m - k - l}^\infty \mspace{-12.mu}
		                           n J_n (- \i \tau) \, \i^{-n} \biggr\} \epp
\end{multline}
It is straightforward to reduce this to equations (\ref{alkmodbesselupper}),
(\ref{alkmodbessellower}) of the main text using (\ref{besselgenfun}) as well
as the identities
\begin{equation}
     J_{-n} (z) = (-1)^n J_n (z) \epc \qd
     2 n J_n (z) = z \bigl( J_{n-1} (z) + J_{n+1} (z)\bigr)
\end{equation}
and the definition
\begin{equation}
     I_n (z) = \i^{-n} J_n (\i z)
\end{equation}
of the modified Bessel functions of the first kind.
\subsection{Well-definedness}
\label{app:welldef}
We shall prove below that
\begin{equation} \label{pqsolvability}
     \det_{k, l = 0, \dots, n} A_k^l (m) \bigr|_{\tau = 0} \ne 0
\end{equation}
for $n = m - 1, m$ and $0 < h < 4J$. This implies that the functions
\begin{equation}
     t \mapsto \det_{k, l = 0, \dots, n} A_k^l (m)
\end{equation}
$n = m - 1, m$, are not identically zero. Hence, being entire functions,
they can at most vanish on a discrete subset of ${\mathbb C}$ with an
accumulation point at infinity.

Using that $I_n (0) = \de_{n, 0}$ in (\ref{alkmodbesselupper}) and
(\ref{alkmodbessellower}) we obtain
\begin{equation} \label{aklm}
     A_k^l (m) \bigr|_{t = \frac{\i}{2T}}
        = - \frac{\i^{m - k - l}}{c}
	  \begin{cases}
	     0 & k + l < m - 1 \\ 1 & k + l = m - 1 \\
	     2 - 2 \i c (k + l - m + 1) & k + l > m - 1 \epp
	  \end{cases}
\end{equation}
It follows that
\begin{equation}
     \det_{k, l = 0, \dots, m - 1} A_k^l (m) \bigr|_{t = \frac{\i}{2T}}
        = (\i c)^{- m} (-1)^\frac{m (m - 1)}{2}
\end{equation}
which is clearly non-zero for $0 < h < 4J$.

The evaluation of the other determinant, $n = m$ in equation
(\ref{pqsolvability}), is more tricky. Using (\ref{aklm}) and
elementary row- and column manipulations of the determinant we
can rewrite it as
\begin{equation}
     \det_{k, l = 0, \dots, m} A_k^l (m) \bigr|_{t = \frac{\i}{2T}}
        = (- c)^{- m - 1} (-1)^\frac{m (m + 1)}{2} \det U
\end{equation}
where $U$ is the $(m + 1) \times (m + 1)$ tridiagonal matrix
\begin{equation}
     U = \begin{pmatrix}
	    1 - 2 \i c & 1 &&&& \\
	    - 1 & - 2 \i c & 1 &&& \\
	    & - 1 & - 2 \i c & 1 & \\
	    &&& \ddots  & & \\
	    &&& -1 & - 2 \i c & 1 \\
	    &&&& - 1 & 1 - 2 \i c
         \end{pmatrix}
\end{equation}
Clearly $\det U$ is a polynomial in $c$ of degree $m + 1$ with highest coefficient
$(- 2 \i)^{m+1}$. We shall show that this polynomial is non-zero for
$0 < h < 4J$.

For every root, $\det U = 0$, there is an $\xv = (x_0, \dots, x_m)^t
\in {\mathbb C}^{m+1}$, $\xv \ne 0$, such that $U \xv = 0$ or,
equivalently,
\begin{equation}
\begin{split} \label{kerneluxn}
     & x_1 + (1 - 2 \i c) x_0 = 0 \epc \\
     & x_{n+1} - 2 \i c x_n - x_{n-1} = 0 \epc \qd n = 1, \dots, m \epc \\
     & x_{m+1} = x_m \epp
\end{split}
\end{equation}
Hence, $x_0 = 0$ implies $\xv = 0$, and we must have $x_0 \ne 0$. Let
\begin{equation}
     T = \begin{pmatrix} 0 & 1 \\ 1 & 2 \i c \end{pmatrix} \epp
\end{equation}
Then it follows from (\ref{kerneluxn}) that
\begin{equation}
     x_m \begin{pmatrix} 1 \\ 1 \end{pmatrix}
        = T^m \begin{pmatrix} 1 \\ 2 \i c - 1 \end{pmatrix} x_0 \epp
\end{equation}
Since $x_0 \ne 0$, we obtain the following necessary condition for
$c$ to be a root of $\det U$,
\begin{equation} \label{neccondc}
     (1, - 1) \, T^m \begin{pmatrix} 1 \\ 2 \i c - 1 \end{pmatrix} = 0 \epp
\end{equation}
The matrix $T$ is non-degenerate with eigenvalues $\la_\pm = \i \re^{\mp \i p_F}$
and corresponding eigenvectors $\vv_\pm = (1, \la_\pm)^t$ if $0 < h < 4J$.
Moreover,
\begin{equation}
     \begin{pmatrix} 1 \\ 2 \i c - 1 \end{pmatrix}
        = \frac{\i \re^{- \i p_F} - 1}{2 \sin (p_F)} \vv_+
	  - \frac{\i \re^{\i p_F} - 1}{2 \sin (p_F)} \vv_- \epp
\end{equation}
Inserting the latter into (\ref{neccondc}) and using the explicit
form of $\la_\pm$ and $\vv_\pm$ we end up with
\begin{equation}
     (1, - 1) \, T^m \begin{pmatrix} 1 \\ 2 \i c - 1 \end{pmatrix}
        = 2 \i^m \biggl[\frac{\sin\bigl((m+1) p_F\bigr)}{\sin(p_F)}
	                - \i \cos \bigl((m+1) p_F\bigr)\biggr] \epp
\end{equation}
The right hand side never vanishes for $p_F \in (0,\p/2)$, since sine
and cosine are both real and do not have common zeros. This entails the
claim.
\subsection{Behaviour at large negative times}
\label{app:alklarget}
\begin{proposition}
For $t \rightarrow - \infty$ the polynomials $P_m$ and $Q_{m+1}$
behave as
\begin{equation} \label{largetpq}
     P_m (z) = (z - \i)^m + {\cal O} \bigl( t^{-1} \bigr) \epc \qd
     Q_{m+1} (z) = (z - \i)^{m+1} + {\cal O} \bigl( t^{-1} \bigr) \epp
\end{equation}
\end{proposition}
\begin{proof}

We shall analyse the integrals (\ref{alkrationalint}) for $\tau \rightarrow
+ \infty$ ($t \rightarrow - \infty$) by means of the steepest-descent method.
This will be enough to understand the behaviour of the correlation function
for $t \rightarrow \pm \infty$, since $\<\s_1^- \s_{m+1}^+ (t)\>_T =
\<\s_1^- \s_{m+1}^+ (-t) \>_T^\ast$. The reason why $t \rightarrow - \infty$
is easier to analyse than $t \rightarrow + \infty$ is that in the latter
case the relevant saddle point coincides with the pole of $b_m$ at
$w = - \i$. In fact, the two saddle points, $f' (w) = 0$, are at $w = \pm \i$,
where $f(\pm \i) = \pm 1$. Hence, the saddle point at $+ \i$ is dominant for
$\tau \rightarrow + \infty$, while the saddle point at $- \i$ dominates
for $\tau \rightarrow - \infty$. The saddle-point contours are easily
determined in this case. They consist of the unit circle plus the imaginary
axis. The steepest descent and steepest ascent directions close to the
saddle points are indicated in figure~\ref{fig:alk_saddle_point}.
\begin{figure}
\begin{center}
\includegraphics[width=.70\textwidth]{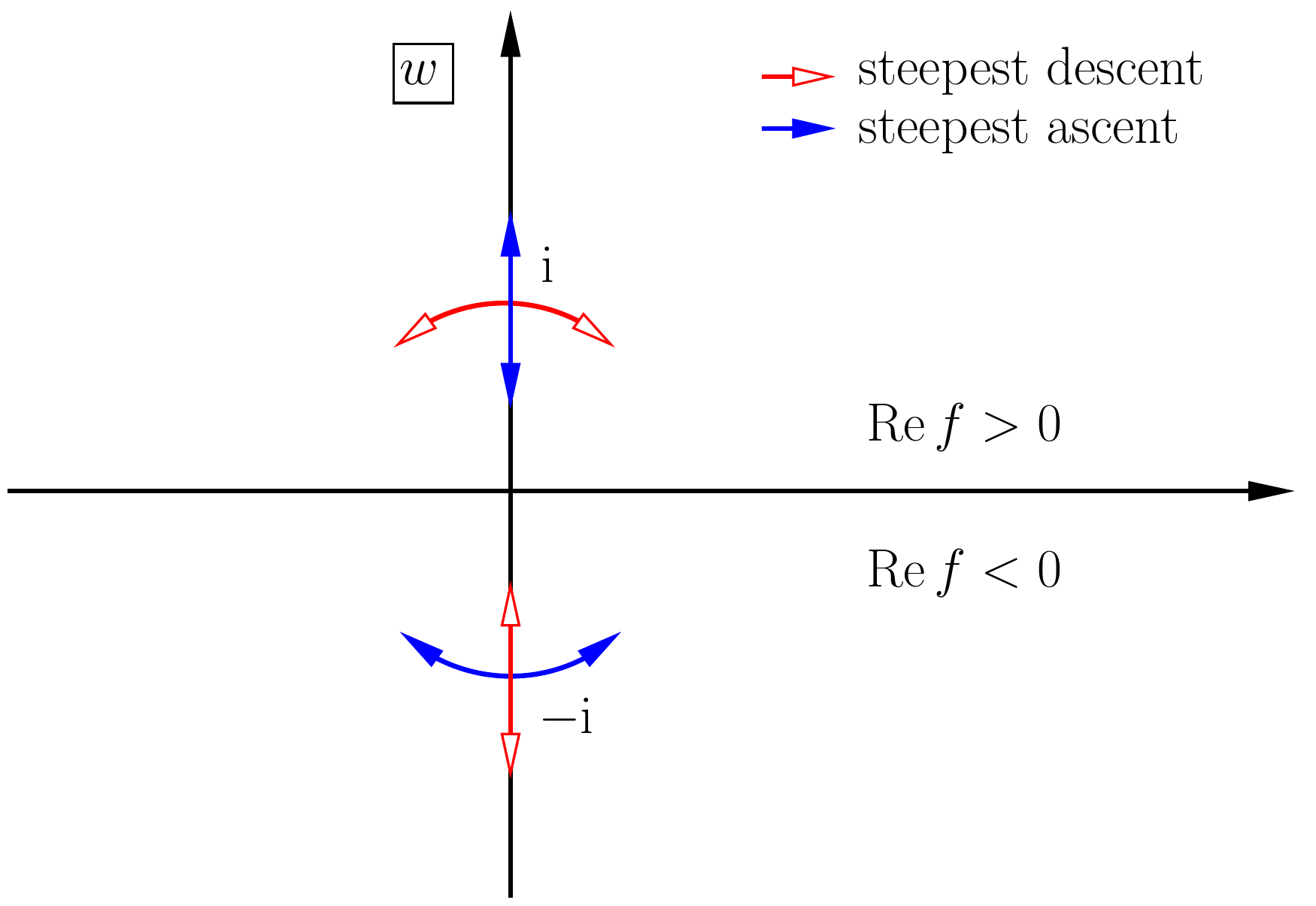}
\caption{\label{fig:alk_saddle_point} Sketch of saddle points and
saddle point contours.
}
\end{center}
\end{figure}

According to the usual reasoning of the steepest descent method we can restrict
the integration contour to a vicinity of the saddle point at $+ \i$. We may,
for instance, choose a semi-circle of unit radius in the upper half plane around
the origin, explicitly parameterized as $\th \mapsto w = \i \re^{\i \th}$,
$\th \in [- \p/2, \p/2]$. Then
\begin{equation} \label{alksp1}
     A^l_k (m) \re^{- \tau} =
        \int_{-\p/2}^{\p/2} \rd \th \: \re^{\i \th} b_m \bigl(\i \re^{\i \th}\bigr)
	                               \bigl(\i \re^{\i \th}\bigr)^{k + l} \re^{- 2 \tau \sin^2 (\th/2)}
				       + {\cal O} \bigl(\re^{- \tau}\bigr) \epp
\end{equation}
Setting
\begin{equation}
     h(\th) = \sqrt{2} \sin (\th/2) \epc \qd
     \fb_m (u) = \frac{\re^{\i h^{-1} (u)} b_m \bigl(\i \re^{\i h^{-1} (u)}\bigr)}{h' \circ h^{-1} (u)}
\end{equation}
and substituting $u = h(\th)$ as integration variable in (\ref{alksp1})
we obtain an all-order asymptotic expansion of the integral on the
right hand side of (\ref{alksp1}),
\begin{equation}
     A^l_k (m) \re^{- \tau} = \i^{k + l}
        \sum_{p \ge 0} \frac{\G(p + \2)}{\tau^{p + \2} (2p!)}
	   \Bigl[\6_u^{2p} \fb_m (u) \re^{\i (k + l) h^{-1} (u)} \Bigr]_{u = 0}
	   + {\cal O} \bigl(\re^{- \tau}\bigr) \epp
\end{equation}

Inserting this into the determinant in the denominator of (\ref{dcramer}) and
using the multi-linearity of the determinant we obtain
\begin{multline} \label{detalkop}
     \det_{k, l = 0, \dots, m} \bigl\{ A^l_k (m) \bigr\} \sim
	\re^{(m+1) \tau} (-1)^\frac{m(m+1)}{2}
	\sum_{p_0, \dots, p_m \ge 0}
	\biggl[ \prod_{r=0}^m \frac{\G(p_r + \2)}{\tau^{p_r + \2} (2p_r!)} \6_{u_r}^{2p_r}\biggr] \\
	\times \biggl( \prod_{s=0}^m \fb_m (u_s) \re^{\i s \cdot h^{-1} (u_s)} \biggr)
     \det_{k, l = 0, \dots, m} \bigl\{ \re^{\i k h^{-1} (u_l)} \bigr\} \Bigr|_{u_j = 0} \epc
\end{multline}
where ``$\sim$'' means asymptotically equal. The operator
\begin{equation}
     \sum_{p_0, \dots, p_m \ge 0}
     \biggl[ \prod_{r=0}^m \frac{\G(p_r + \2)}{\tau^{p_r + \2} (2p_r!)}
             \6_{u_r}^{2p_r}\biggr]_{u_j \rightarrow 0}
\end{equation}
sends functions $\Ph (u_0, \dots, u_m)$ that are antisymmetric in any pair
of variables $u_j$, $u_k$ to zero. We may therefore replace the term in
the second line of (\ref{detalkop}) by its symmetrized version,
\begin{multline} \label{detalkopsym}
     \det_{k, l = 0, \dots, m} \bigl\{ A^l_k (m) \bigr\} \sim
	\re^{(m+1) \tau} (-1)^\frac{m(m+1)}{2} \frac{1}{(m+1)!} \\
	\sum_{p_0, \dots, p_m \ge 0}
	\biggl[ \prod_{r=0}^m \frac{\G(p_r + \2)}{\tau^{p_r + \2} (2p_r!)} \6_{u_r}^{2p_r}\biggr]
	\biggl[ \prod_{s=0}^m \fb_m (u_s) \biggr]
        \Bigl( \det_{k, l = 0, \dots, m} \bigl\{ \re^{\i k h^{-1} (u_l)} \bigr\}
	        \Bigr)^2 \Bigr|_{u_j = 0} \epp
\end{multline}

>From this expression we want to extract the leading term in $\tau$. Many terms
under the sum on the right hand side vanish, e.g.\ the term $p_0 = p_1 = \dots =
p_m$, because the determinants vanish. Non-vanishing terms are generated by
the action of the derivatives on the columns of the determinants. The determinants
are non-vanishing only if the degrees of the derivatives inside the columns are
mutually different. The term of lowest possible degree of the derivatives is
generated by $\prod_{r=0}^m \6_{u_r}^{2r}$ and corresponds to summands with
$p_r = r$; $r = 0, \dots, m$, modulo permutations. This term can be calculated
explicitly,
\begin{multline} \label{leadingder}
     \biggl[ \prod_{r=0}^m \6_{u_r}^{2p_r}\biggr]
	\biggl[ \prod_{s=0}^m \fb_m (u_s) \biggr]
        \Bigl( \det_{k, l = 0, \dots, m} \bigl\{ \re^{\i k h^{-1} (u_l)} \bigr\}
	        \Bigr)^2 \Bigr|_{u_j = 0} \\ =
        \biggl[ \prod_{r=0}^m \frac{(2r)! \fb_m (0) \i^m}{(r!)^2 (h' (0))^m} \biggr]
        \Bigl( \det_{k, l = 1, \dots, m} \bigl\{ k^l \bigr\} \Bigr)^2 \epp
\end{multline}
Here the combinatorial factor $\prod_{r=0}^m (2r)!/(r!)^2$ comes from the
application of the Leibniz rule. Since the degrees of the derivatives in
(\ref{detalkopsym}) are connected with the powers of $\tau$, (\ref{leadingder})
corresponds to the leading asymptotics, and we conclude that
\begin{multline} \label{detalkasy}
     \det_{k, l = 0, \dots, m} \bigl\{ A^l_k (m) \bigr\} =
	\biggl[ \prod_{r=0}^m \frac{\re^\tau \G(r + \2) \fb_m (0)}
	                           {\tau^{r + \2} (r!)^2 (h' (0))^m}\biggr]
        \Bigl( \det_{k, l = 1, \dots, m} \bigl\{ k^l \bigr\} \Bigr)^2
	\bigl(1 + {\cal O} (\tau^{-1})\bigr) \\ =
	\frac{G(m + \frac{3}{2}) \bigl( - \i \sqrt{2} \bigr)^{m^2 - 1} \re^{(m+1)(1 + \i c) \tau}}
	     {G(\2) (2 \p)^{m+1} \tau^\frac{(m+1)^2}{2}}
	\bigl(1 + {\cal O} (\tau^{-1})\bigr) \epc
\end{multline}
where $G$ is the Barnes function.

For the numerator in (\ref{dcramer}) we have to replace in the $n$th column of
the determinant $\Dv_n \rightarrow \Dv_{m+1}$. This amounts to replacing the
column index $k$ in (\ref{alksp1}) and in the following equations by
\begin{equation}
     x_k^{(n)} = k + \de_{n, k} (m + 1 - n) \epp
\end{equation}
The analysis stays very similar with this minor modification. After a few steps
we obtain
\begin{multline} \label{detalkmodnasy}
     \det_{k, l = 0, \dots, m} \Bigl\{ A^l_{x_k^{(n)}} (m) \Bigr\} =
	\i^{m + 1 - n}
	\biggl[ \prod_{r=0}^m \frac{\re^\tau \G(r + \2) \fb_m (0)}
	                           {\tau^{r + \2} (r!)^2 (h' (0))^m}\biggr]
        \det_{k, l = 1, \dots, m} \bigl\{ k^l \bigr\} \\[1ex] \times
	\det_{k, l = 0, \dots, m} \bigl\{- \i h'(0) \6_{u_l}^l \re^{\i x_k^{(n)} h^{-1} (u_l)}
	                                   \big|_{u_l = 0} \bigr\}
	\bigl(1 + {\cal O} (\tau^{-1})\bigr) \epp
\end{multline}
The ratio between (\ref{detalkasy}) and (\ref{detalkmodnasy}) can be easily
calculated explicitly. Then, using (\ref{dcramer}), we end up with
\begin{equation}
     d_{m+1}^{(n)} = (- \i)^{m + 1 - n} \binom{m+1}{n}
                     \bigl(1 + {\cal O} (\tau^{-1})\bigr) \epp
\end{equation}
Inserting this into (\ref{pqcoeff}) we obtain the asymptotic formula
(\ref{largetpq}) for $Q_{m+1}$. The derivation of the large-$\tau$
asymptotics of $P_m$ is similar.
\end{proof}
\begin{corollary} For large negative times the `prefactor' in the
asymptotic formula (\ref{xxtransht}) for the transverse correlation
function behaves as
\begin{equation}
     Q_{m+1}(- \i) P_m' (- \i) - P_m (- \i) Q_{m+1}' (- \i) =
        (-1)^{m+1} 4^m \bigl(1 + {\cal O} (t^{-1})\bigr) \epp
\end{equation}
\end{corollary}
The long-time behaviour of the function $u_m$ appearing in the exponent
in (\ref{xxtransht}) is harder to estimate, since the asymptotic forms
of the polynomials $P_m$ and $Q_{m+1}$ have high-order zeros at $w = \i$.

\section{\boldmath Explicit results for small $m$}
\label{app:explicit_m} \vspace{-1ex} \noindent
We set $\tilde{\tau}=-4Jt$.  Even the first few coefficients of the 
polynomials $P_m(z)$ and $Q_m(z)$ are already too lengthy to be reproduced
here. On the other hand, the combination  $W_m(\tilde{\tau}) :=
\bigl(Q_{m+1}(-i) P_m'(-i) - Q'_{m+1}(-i) P_m(-i) \bigr) \,(m=1,2)$
turns out to be relatively simple. We write 
\begin{equation}
     {W}_m(\tilde{\tau}) = \frac{N_m (\tilde{\tau}, h)}{D_m( \tilde{\tau}, h)}
                           \epc \qd m = 1, 2 \epp
\end{equation}
The denominators and numerators are then explicitly given by
\begin{subequations}
\begin{align}
     N_1( \tilde{\tau}, h) & = 2 e^{\tilde{\tau}}(4 i J \tilde{\tau}^{-1}+ h)  I_1(\tilde{\tau}) \epc \\[1ex]
     D_1( \tilde{\tau}, h) & =
        - 4 i J + h  \tau e^{\tilde{\tau}} \bigl(I_0( \tilde{\tau})-I_1(\tilde{\tau}) \bigr) \epc \\[1ex]
     N_2(\tilde{\tau}, h) & = \tilde{\tau} ^2 \Bigr(e^{2 \tilde{\tau} } I_0(\tilde{\tau} ){}^2
        \left(h^2 (2 \tilde{\tau} -1)-16 J^2\right) \notag \\
	& \mspace{72.mu} + 2 e^{2 \tilde{\tau} } I_1(\tilde{\tau} ) I_0(\tilde{\tau} )
	  \left(16 J^2-h^2 \tilde{\tau} \right)
	  + \left(h-4 i J e^{\tilde{\tau} } I_1(\tilde{\tau} )\right){}^2 \Bigr) \epc \\[1ex]
     D_2( \tilde{\tau}, h) & = 4 e^{2 \tilde{\tau} }
          \Bigl(h \tilde{\tau}  I_0(\tilde{\tau} ){}^2 (h \tilde{\tau} +8 i J) \notag \\
	  & \mspace{108.mu} - 2 h I_1(\tilde{\tau} ) I_0(\tilde{\tau} ) (h \tilde{\tau} +8 i J)
            + I_1(\tilde{\tau} ){}^2 (4 J-i h \tilde{\tau} )^2\Bigr) \epp
\end{align}
\end{subequations}

The final pieces $u_m\,(m=1,2)$ are also represented in simple forms,
\begin{subequations}
\begin{align}
     u_1 & = 4J - 4J \frac{\partial}{\partial \tau} \ln D_1( \tau, h) \bigr|_{\tau=  \tilde{\tau} }
             \epc \\[1ex]
     u_2 & = -8J \frac{1+\tilde{\tau}}{\tilde{\tau}} -4J
                 \frac{\partial}{\partial \tau} \ln \Bigl( D_2({\tau}, h) \Bigr)
		 \Bigr|_{\tau=\tilde{\tau}} \epp
\end{align}
\end{subequations}
By substituting these into the formula (\ref{xxtransht}), we obtain
\begin{subequations}
\begin{align}
     \bigl\< \sigma_1^- \sigma_2^+ (t) \bigr\>_{T}  & \sim  
        - \frac{J}{T} {\rm e}^{-4J^2 t^2 -ih t }
	  \frac{4 i J + h \tilde{\tau}  }{4i J \tilde{\tau}} I_1(\tilde{\tau}) \epc \\[1ex]
     \bigl\< \sigma_1^- \sigma_3^+ (t) \bigr\>_{T} & \sim
          \frac{1}{2} \Bigl( -\frac{J}{T}\Bigr)^2  \frac{{\rm e}^{-4J^2 t^2 -ih t}}{4 J^2 \tilde{\tau} ^2}
          \bigl( h \tilde{\tau} I_0(\tilde{\tau}){}^2 (h \tilde{\tau} +8 i J) \notag \\
	  & \mspace{90.mu} -2 h I_1(\tilde{\tau} ) I_0(\tilde{\tau} ) (h \tilde{\tau} +8 i J)
	  +I_1(\tilde{\tau} ){}^2 (4 J-i h \tilde{\tau} )^2 \bigr) \epp
\end{align}
\end{subequations}
When $h=0$, these expressions reduce to
\begin{subequations}
\begin{align}
     \bigl\<\sigma_1^- \sigma_2^+ (t) \bigr\>_{T, h=0}
        & \sim- \frac{J}{T} {\rm e}^{-4J^2 t^2}  \frac{I_1(\tilde{\tau})}{\tilde{\tau}}, \\[1ex]
     \bigl\< \sigma_1^- \sigma_3^+ (t) \bigr\>_{T, h=0} & \sim  2 \Bigl(\frac{J}{T}\Bigr)^2  {\rm e}^{-4J^2 t^2} 
          \Bigl( \frac{I_1(\tilde{\tau})}{\tilde{\tau}} \Bigr)^2, 
\end{align}
\end{subequations}
and the leading order terms in \cite{PeCa77} (eq.\ (6.36)) are recovered if we identify
$J_{\text{Perk-Capel}}=2 J$.

For larger $m$, we still have difficulties in  manipulating huge expressions and present
only the result for $m=3$ and $h=0$,
\begin{subequations}
\begin{align}
     {W}_3(\tilde{\tau})|_{h=0} & =-\frac{8 e^{\tilde{\tau} }
        \left(-\tilde{\tau} ^2 I_0(\tilde{\tau} ){}^2
	+ \tilde{\tau} ^2 I_1(\tilde{\tau} ){}^2
	+ \tilde{\tau}  I_1(\tilde{\tau} ) I_0(\tilde{\tau} )
	+ 2 I_1(\tilde{\tau} ){}^2\right)}{\tilde{\tau} ^3 I_1(\tilde{\tau})} \epc \\[1ex]
     u_3|_{h=0} & = 4J  \frac{\partial}{\partial \tau}
        \Bigl[ \tau+2  \ln \bigl( \frac{\tau}{I_1(\tau)} \bigr)  \Bigr]_{\tau=\tilde{\tau}} \epc \\[1ex]
     \bigl< \sigma_1^-  \sigma_4^+ (t) \bigr\>_{T, h=0}  &\sim 16 \Bigl( -\frac{J}{T} \Bigr)^3
          {\rm e}^{-4J^2 t^2} \frac{I_1(\tilde{\tau})}{ \tilde{\tau}^5} \notag \\
	  & \mspace{72.mu} \times
	    \Bigl( - \tilde{\tau}^2 I_0( \tilde{\tau})^2
	           + \tilde{\tau}  I_0( \tilde{\tau}) I_1( \tilde{\tau})
		   + (\tilde{\tau}^2+2) I_1( \tilde{\tau})^2 \Bigr) \epp
\end{align}
\end{subequations}
This is (\ref{examplem4}) of the main text.

}


\providecommand{\bysame}{\leavevmode\hbox to3em{\hrulefill}\thinspace}
\providecommand{\MR}{\relax\ifhmode\unskip\space\fi MR }
\providecommand{\MRhref}[2]{%
  \href{http://www.ams.org/mathscinet-getitem?mr=#1}{#2}
}
\providecommand{\href}[2]{#2}

\end{document}